\documentclass{article}

\usepackage{xspace}
\usepackage{cite}
\usepackage{amsmath,amssymb,amsfonts,amsthm}
\usepackage{graphicx}
\usepackage{textcomp}
\usepackage{xcolor}

\usepackage{mdframed}
\usepackage{adjustbox,float}
\usepackage[boxed,vlined,ruled,linesnumbered]{algorithm2e}
\usepackage{url}
\usepackage{hyperref}
\usepackage{upgreek}

\newcommand{\remove}[1]{}
\newcommand{\myboldmath}{}
\newcommand{\defn}[1]{{\textit{\textbf{\myboldmath #1}}}}

\newcommand{\poly}{\mathrm{poly}}
\newtheorem{definition}{Definition}
\newtheorem{theorem}{Theorem}
\newtheorem{lemma}{Lemma}

\newcommand{\cM}{{\mathcal M}}

\newcommand{\cA}{{\mathcal A}}

\newcommand{\cE}{{\mathcal E}}
\newcommand{\cH}{{\mathcal H}}

\newcommand{\cP}{{\mathcal P}}

\newcommand{\mwd}{MLI\xspace}
\newcommand{\mwdfull}{Maximum Local Interference\xspace}
\newcommand{\mwdfulllong}{Maximum Local Interference on Heavy Links\xspace}
\newcommand{\probfull}{Local Broadcast\xspace}
\newcommand{\probshort}{LB\xspace}

\newcommand{\received}{received\xspace}
\newcommand{\weight}{weight\xspace}
\newcommand{\heavy}{heavy\xspace}
\newcommand{\round}{step\xspace}
\newcommand{\rounds}{steps\xspace}

\newcommand{\parhead}[1]{\noindent{\textbf{#1.}\xspace}}

\newcommand{\LS}{Link Scheduling\xspace}

\newcommand{\congest}{{\fontfamily{cmss}\selectfont CONGEST}\xspace}

\newcommand{\dk}[1]{{\color{red}{#1}}}
\newcommand{\mm}[1]{{\color{cyan}{#1}}}
\newcommand{\mam}[1]{{\color{cyan}{#1}}}
\newcommand{\mamr}[1]{{\color{cyan}{#1}}}
\newcommand{\mmr}[1]{{\color{cyan}{#1}}}

\newcommand{\miguel}[1]{{\color{cyan}{#1}}}
\newcommand{\mgl}[1]{{\color{cyan}{#1}}}
\newcommand{\mglast}[1]{{\color{blue}{#1}}}
\renewcommand{\dk}[1]{#1}
\renewcommand{\mm}[1]{#1}
\renewcommand{\mam}[1]{#1}
\renewcommand{\mamr}[1]{#1}
\renewcommand{\mmr}[1]{#1}

\renewcommand{\miguel}[1]{#1}
\renewcommand{\mgl}[1]{#1}
\renewcommand{\mglast}[1]{#1}

\usepackage{fullpage}
\usepackage{array}
\newcolumntype{x}[1]{>{\centering\arraybackslash\hspace{0pt}}p{#1}}

\usepackage{graphicx}

\begin{document}

\title{Weighted Beeping Networks~\thanks{A preliminary version of this work has appeared in~\cite{discversion}.}
}

\author{
Dariusz~R.~Kowalski~\thanks{School of Computer and Cyber Sciences, Augusta University}\and
Miguel~A.~Mosteiro~\thanks{Computer Science Department, Pace University}
}

\date{}

\maketitle 

\begin{abstract}

\dk{Beeping Networks (BN) have become an established model of wireless communication and computation. They assume, however, that the communication interference caused by beeping nodes is binary (OR channel), while both artificial and nature-driven networks are based on varied signal strength and receiving thresholds. This work naturally extends the classic beeping networks to weighted beeping networks (WBN), by modeling signal strengths by weights on graph edges and receiving a beep if the accumulative signal strengths of transmitting neighbors are above receiving threshold. Under this generalized model, we study the fundamental \probfull (\probshort) problem. Since edges with low weights may not even accumulate to a received beep, we only require that messages are delivered along sufficiently heavy edges, with weights above some given threshold $\sigma$. We design an \probshort algorithm that completes its task in $O(I_\sigma \log n \cdot (\mathcal{M}+\log n))$ beeping rounds, whp, where $n$ is the number of nodes, $I_\sigma$ is the maximum local interference on heavy edges, and $\mathcal{M}$ is the maximum length of messages sent locally. This result is nearly optimal, as, with probability~$1$, $\Omega(I_\sigma)$~beeping rounds are necessary to solve \probshort. Finally, we show that any algorithm designed for the classic beeping networks and terminating in polynomial number of beeping rounds, can be emulated on any WBN with a logarithmic overhead, that is, it returns the same output as working on the classic beeping network isomorphic to the sub-graph of heavy edges.}

{\bf Keywords:} \probfull, Multiagent Systems, Weighted Beeping Networks.
\end{abstract}


\section{Introduction}
\label{s:intro}

In this work, we study fundamental communication and computation problems in multi-agent wireless networks embedded in harshly constrained environments. Smart agent systems—such as insect colonies, bird flocks, and bacteria populations—rely on the collective intelligence of individual organisms communicating via simple wireless signals (see the recent survey by~\cite{liu2024review}). To study the capabilities of these systems, distributed computing has heavily relied on the Beeping Network (BN) model~\cite{cornejo2010deploying,afek2011biological,afek2013beeping}. In standard BNs, nodes can only choose to beep or listen in discrete time slots. If a node listens, it receives a binary response over an OR channel: it hears a beep if at least one neighbor transmits, and silence otherwise. (An overview of the various BN model extensions and problems studied in the literature can be found in Section~\ref{sec:motivrelated}.)

While the classic binary BN model provides a minimalistic abstraction, it is fundamentally insufficient for capturing real-world biological and physical phenomena. Physical channels, artificial or nature-driven, are rarely binary. Instead, they exhibit varied signal strengths and receiving thresholds. To address this gap, we naturally extend the classic beeping model to Weighted Beeping Networks (WBN). In a WBN, signal strengths are explicitly modeled by weights on graph edges. Rather than an OR channel, communication is governed by a threshold channel: a listening node receives a beep if and only if the cumulative signal strengths of its transmitting neighbors exceed a designated receiving threshold (normalized~to~1).\footnote{Note that the threshold channel becomes an OR channel, and WBN collapses to BN, if all positive~weights~are $1$.}

Under this newly generalized model, we investigate the fundamental \probfull (\probshort) problem, where nodes seek to propagate local information to their immediate vicinity. (See~\cite{garncarek2025deterministic,lynch2015truly,halldorsson2012towards,Hay} and the references therein.) 
Moving to a weighted model introduces unique algorithmic hurdles: links with low weights may fail to accumulate enough signal strength to overcome the receiving threshold, rendering them effectively silent. Consequently, we focus on a localized variation of the problem where messages must be successfully delivered along sufficiently heavy edges—defined as links with weights above a given parametric threshold $\sigma$.

Resolving \probshort in WBNs is highly non-trivial. Because the discrete behaviors of neighboring nodes are compressed into a single aggregated bit at the threshold channel, algorithms must carefully schedule transmissions to prevent heavy signals from being drowned out by surrounding interference. We overcome these challenges by characterizing the complexity of the problem using a novel topological metric: the maximum local interference on heavy edges, denoted as $I_\sigma$.

\subsection{Our Contributions}

\begin{itemize}
    \item \underline{Algorithmic Contribution}: We design a \probfull algorithm tailored for WBNs that successfully completes its task with high probability (whp)\footnote{For a network size $n$, we say that an event $E$ occurs \defn{with high probability} if $Pr(\lnot E)\leq n^{c}$, for some constant $c<0$.} in $O(I_\sigma \log n (\cM+\log n))$ beeping rounds, where $n$ is the total number of nodes, $\cM$ is the maximum length of locally transmitted messages, and $I_\sigma$ is the maximum local interference on heavy edges.
    \item \underline{Complexity Upper/Lower Bounds}: We show that our algorithmic result is nearly optimal. We establish an almost tight lower bound proving that, with probability 1, $\Omega(I_\sigma)$ beeping rounds are strictly necessary to solve the \probfull problem in this environment.
    \item \underline{General Emulation Theorem}: Beyond standalone algorithms, we prove a powerful structural relationship between classic and weighted beeping networks. We demonstrate that any algorithm originally designed for classic BNs terminating in a polynomial number of rounds can be emulated on any WBN with only a {\bf\it logarithmic overhead}. The emulated protocol is guaranteed to produce the exact same output as if it were executing directly on a classic, unweighted beeping network isomorphic to the sub-graph of heavy~edges~($E_\sigma$).
\end{itemize}

\textbf{Details of Our Contributions.}
More specifically, we study \probshort in a system of $n$ autonomous nodes under the WBN model, where each agent is installed in one network node.
We assume that nodes beep synchronously in {\it rounds}, and we measure algorithmic performance in beeping rounds.
We study provably efficient algorithms,
run autonomously by nodes in WBL networks.

We identify a characteristic of WBNs, which we call \mwdfulllong, also called \mwdfull (\mwd) for short, and denote as $I_\sigma$, for a parametric threshold $\sigma$.\footnote{The \probshort problem, similarly to almost all complex graph problems, does not have an exact formula for the number of rounds needed to solve the problem, only an approximation as a function of some graph characteristic that holds asymptotically.
E.g., time-related parameters (mixing, covering, etc.) of random walks are only approximated by concise formulas, such as network diameter, vertex/edge expansion or conductance. 
We follow the same methodology~here.}
We show mathematically~that 

\vspace*{2ex}
\hspace*{-1.8em}
\defn{$I_\sigma$ is a good approximation of the minimum number of rounds needed to solve the \probfull problem,}

\vspace*{2ex}
\noindent
by showing that the ratio between the upper and lower bound is in $O(\log^2 n)$.
More precisely, we show in Section~\ref{s:upbound} that our \probshort algorithm 
locally propagates messages of length $\cM$ in $O(I_\sigma \log n \cdot (\cM+\log n))$ rounds on {\it any} topology graph $G$ with \mwd at most $I_\sigma$ (see Theorem~\ref{thm:algorithm-LB}). (Note that the additive $\log n$ comes from attaching the node identifier to the message, and it subsumes $\cM$ for short messages, i.e., $\cM=O(\log n)$.) 
This bound holds with high probability, i.e., with error polynomially small in~$n$. 

For clarity, we present and analyze our algorithm in three steps. For settings where only $I_\sigma$ and a polynomial upper bound on $n$ are known and all nodes start execution simultaneously, we present a simplified version of \probshort, called {\it One-beep Local Broadcast (OBLB)}, in which, for each node, it is enough to guarantee the occurrence of at least one ``dominant'' beep from each \heavy neighbor, that is, a beep from an \heavy neighbor when the aggregated \weight of all other beeping neighbors is less than the 
hearing threshold $1$ 
(Sec.~\ref{sec:one-beep}). Note that the WBN model does not provide a mechanism to detect a dominant beep.
Next, we generalize our techniques to \probshort (Sec.~\ref{sec:general-algorithm}), by incorporating algorithmic techniques (based on short codes) that allow to identify each occurrence of dominant beep in the simulated execution of our OBLB algorithm. 
Finally, we show how to remove the knowledge of $I_\sigma$ and synchronization assumptions, 
and discuss the case of dynamic weights 
(Sec.~\ref{sec:extensions}).

In Section~\ref{sec:lower-bound}, we prove an almost matching lower bound of $\Omega(I_\sigma)$ for OBLB (see Theorem~\ref{thm:lower}).
This lower bound
automatically holds for \probshort, as a more general problem.
To prove the lower bound, we show an adversarial input on which any OBLB algorithm requires the claimed time. Motivated by the topology of classic neural networks, we used a bipartite graph with certain weight function -- we conjecture that our approach could be further extended to other more specific network classes, such as multi-layered networks\footnote{%
Multi-layered WBNs could be seen as beeping neural networks, similarly to spiking neural networks.} 
or scale-free graphs. 
In fact, our extensive experimental study 
(Section~\ref{sec:experimental}) 
strongly indicates that such lower bound applies to a variety of real-world and (commonly used) synthetic topologies (e.g., grids or scale-free graphs).

\mgl{In Section~\ref{sec:general-emulation}, we also present an emulator of any \dk{polynomial-round} algorithm designed for a BN, that runs on a WBN. Our emulator succeeds whp introducing only a $O(\log n)$ time overhead over the original BN algorithm. As a consequence, the wealth of BN algorithms in the literature can be automatically applied to WBN with only a logarithmic overhead.} 


We complement our theoretical study 
with a thorough experimental evaluation 
(Section~\ref{sec:experimental}) 
of our algorithms on two  
environments: social networks in the natural world  (\mmr{four} datasets), on top of which nano-devices could be deployed, and ad-hoc infrastructures (3D-grid, scale-free nets), deployed, e.g.,  
for simultaneous scanning and learning. 
Our experiments show that our theoretical analysis is almost tight and constants hidden in the asymptotic notation are reasonably small in those datasets.
We also show that performance of our algorithm is much better than 
simple/direct approaches
(such as round-robin or periodic transmission schedules with initial random shift).

\mgl{\textbf{Strength of Our Contributions.}
}
Our technical results above have important implications on the efficiency of \probshort.
\mwd is a global characteristic that could be 
\mgl{larger than the number of nodes,}
even if the \heavy neighborhoods (i.e., the sets of \heavy links connected to each receiver) are small.\footnote{E.g., 
each node has only one \heavy link, and all remaining links have \weight just below the threshold.}
Given the locality of the \probshort problem, one would expect to solve it fast in small neighborhoods.  
To the contrary to this intuition,
our lower bound shows that executing \probshort may be costly, even in networks with small degree \dk{of heavy links} (but with high \mwd). 

\mgl{
In recent work~\cite{davies2023optimal}, Davies presented an $\Omega(\Delta \log n)$ lower bound on the overhead of simulating the Broadcast-CONGEST model of communication in BNs, where $\Delta$ is the maximum degree. In a round of the Broadcast-CONGEST model, nodes may send the same $O(\log n)$-bit message to each of their neighboring nodes, and each node hears the messages from all of its neighbors. Thus, the simulation of one Broadcast-CONGEST round is equivalent to solving \probshort with $\cM\in O(\log n)$. Thus, the $\Omega(\Delta \log n)$ lower bound of~\cite{davies2023optimal} applies to \probshort in WBN for the particular case where weights are in $\{0,1\}$ and $\sigma = 1$. In such case our metric becomes $I_\sigma=\Delta-1$ and our upper bound is $O(\Delta \log^2 n)$, yielding the gap a logarithmic factor smaller.

Our lower bound instead applies to the more general case of non-uniform weights and threshold~$\sigma$. Moreover, the adversarial input used in our proof is a geometric network where weights are inversely proportional to distances, resembling the probability function modeling influence in social networks~\cite{kleinberg2000small}.\footnote{We conjecture that our lower bound techniques extend to other weight functions related to social networks, which is left for future work.}
}

\section{Model and Definitions}
\label{s:model}

\subsection{System Model}
We consider a system formed by $n$ autonomous nodes of a Weighted Beeping Network (WBN). 
Nodes have communication and computational capabilities, and each node is labeled with a unique \defn{identifier (ID)} of $O(\log n)$ bits.  
Time is discretized in \defn{rounds}, also called \defn{steps} or \defn{slots} indistinctively.
In each round a node may be either \defn{beeping}, i.e., sending a \defn{beep},\footnote{Some kind of signal, such as radio, light, sound, etc.}
or \defn{listening}.
A round is long enough to beep, or hear a beep if listening (see Section~\ref{sec:communication-model} for details), and possibly perform one bit read/write/compare operation in local memory
which is assumed to take negligible time with respect to the time required by the communication. 
When presenting the main ideas, we assume clock synchronization, but we also show how to get rid of this limitation later 
(Section~\ref{sec:extensions}).

\subsection{WBN Communication Model on General Bidirected~Weighted Graphs}
\label{sec:communication-model}

Due to various reasons (communication interference, preference, etc.), 
the beep of a node may be perceived differently by two other nodes. We model this communication phenomenon with a \defn{\weight} function $w:V\times V\to \mathbb{R}^+ \cup\{0\}$ where 
$V$ is the set of network nodes, and
for any $x,y\in V$, $w(x,x)=0$ and $w(x,y)=w(y,x)$.
The network topology is modeled by a {\bf\it general bidirected weighted graph} $G=(V,E)$, where $E=\{ (x,y) | x,y \in V \land w(x,y)>0\}$ is the set of bidirected links\footnote{Even though the weight is the same in both directions, we use bidirected links to be able to specify the direc\-tion of communication when 
needed. Graph $G$ can be {\bf\it arbitrary} -- weight $0$ means no edge in $E$.}.
An example of a topology graph can be seen in 
Figure~\ref{fig:model}. 
A link from \defn{beeper} $x\in V$ to \defn{listener} $y\in V$ is denoted as $(x,y)$, and the set of nodes beeping at time $t$ is denoted as $V(t)$.
In each time slot, a listening node may \defn{hear}: 
a \defn{beep} 
if the aggregated \weight of the links connecting to beeping neighbors is at least the \defn{hearing threshold}, normalized to $1$, 
or \defn{silence} otherwise (i.e., the total \weight of neighboring beepers is too low to be recognized~by~the~listener).

\begin{figure}[t]
\centering
\begin{minipage}{.5\textwidth}
  \centering
  \includegraphics[width=.9\linewidth]{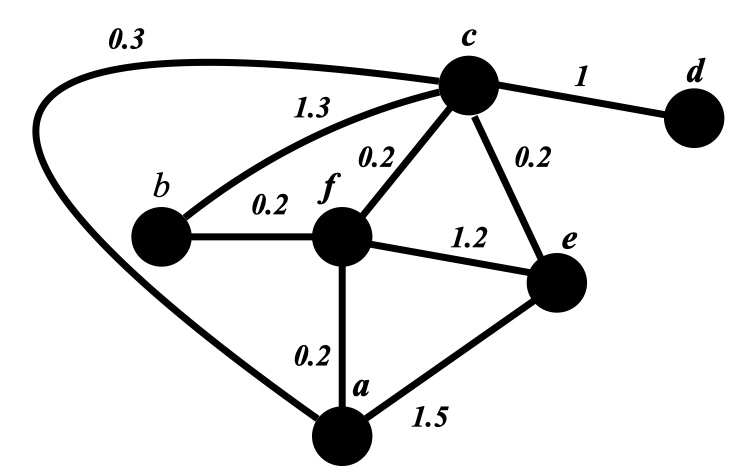}\\(a)
  \label{fig:wbn}
\end{minipage}%
\begin{minipage}{.5\textwidth}
  \centering
  \includegraphics[width=.9\linewidth]{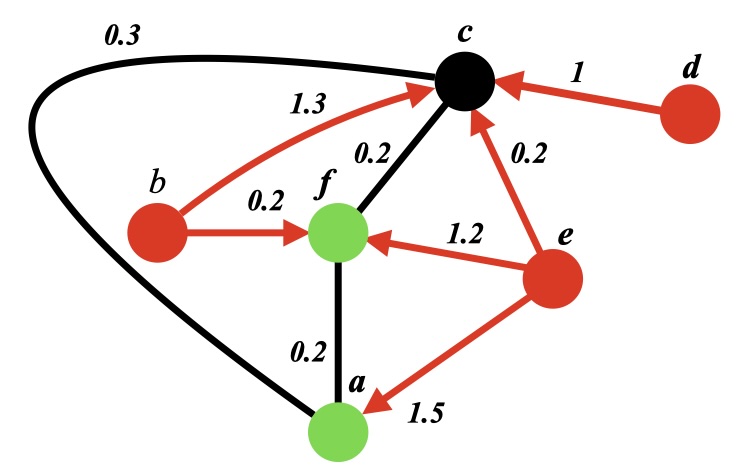}\\(b)
  \label{fig:dominance}
\end{minipage}
\vspace*{-3ex}
\caption{Beeping communication model.
Fig.~(a)~Example of WBN topology modeled by a bidirected weighted graph, see Section~\ref{sec:communication-model}. 
The arrows of each bidirected link are omitted for clarity.
Only edges with positive \weight are depicted.
Fig.~(b)~
Example of the 
dominant beep
concept, as defined in Sec.~\ref{sec:problem}. The depicted arrows point from the beeping nodes to their neighbors, while the remaining arrows are omitted.
Consider a time slot when nodes $b$, $d$, and $e$ beep (in red) and $a$, $f$, and $c$ listen.
Let $\sigma=1$.
Node $a$ hears a beep, its only beeping neighbor is $e$, hence the beep from $e$ is dominant at $a$.
Node $f$ also hears a beep, caused by its two beeping neighbors:~$b$, with \weight $0.2$, and~$e$, with \weight $1.2$; hence, the beep from $e$ is 
dominant at $f$ (indicated in green) while the beep from $b$ is not.
Node $c$ also hears a beep, caused by its three beeping neighbors:~$b$, with \weight $1.3$, node~$e$, with \weight $0.2$, and~$d$, with \weight $1$; 
as a result, no beep 
is dominant~at~$c$.\vspace{-.1in}}
\label{fig:model}
\end{figure}

\subsection{\probfull and Simplified OBLB Problem}
\label{sec:problem}

The following definitions will be used.
For any constant $\sigma\geq 1$, 
called \defn{\weight threshold}, let $E_\sigma\subseteq E$ be the subset of links with \weight at least $\sigma$, that is $E_\sigma=\{(x,y)\in E | w(x,y)\geq \sigma\}$. We call these links \defn{\heavy}, and we extend the notation to nodes saying that $x$ is a \heavy neighbor of $y$ and vice versa.
A \defn{message} is a sequence of bits held by a node for the purpose of sending it to other node(s). 
We say that a node \defn{analyzes} a sequence of beep/silence heard when it decodes such sequence into a sequence of bits, according to a previously established coding system. 
When a node $x$ receives some information from some other node $y$ after analyzing the sequence of beep/silence heard, we say that $x$ has 
\defn{\received}
the information sent by $y$.

\begin{definition}
\label{def:LB}
Given a constant $\sigma\geq 1$ and a WBN with set of \heavy links $E_\sigma$ where each node holds a message, 
the \defn{\probfull (\probshort)} problem is solved in $\Gamma$ rounds if, for each \heavy link $(x,y)\in E_\sigma$, by the end of slot $\Gamma$ node $y$ has \received the ID and message of node $x$ by analyzing the sequence of beep/silence heard.  
\end{definition}
The threshold $\sigma$ restricts learning only to neighbors connected by links in $E_\sigma$. Nevertheless, nodes with lower \weight are still around and may interfere with the reception of information from \heavy neighbors. 

\subsubsection{Tools: beep dominance and OBLB problem.}
In order to design and analyze efficient algorithms for \probfull, we introduce a simpler problem as an algorithmic tool. 
The main technical concept here is the ``dominance'' of a beep 
from node $x$ at node $y$, with $w(x,y)\ge \sigma$, 
which describes a situation in a time slot when node $y$ heard a beep while node $x$ beeped, but also that the total \weight of all the other beepers on $y$ was small.
More formally, for a \weight threshold $\sigma$, we say that a beep from node $x$ 
is \defn{dominant} at node $y$
in some time slot $t$, if the following happens in time slot $t$, 
\begin{quote}
\textbf{\emph{Dominance Condition:}}
$x$ beeps, $y$  listens, $w(x,y)\ge \sigma$,
and the aggregated \weight 
of all other links connecting
beepers in $V(t)$
to $y$
is less than 
the hearing threshold $1$
(i.e. $\sum_{z\in \mm{V(t)}:z\neq x} w(z,y) < 1$).
\end{quote}
Figure~\ref{fig:model} 
illustrates an example of dominance.
Note that, since many beeps or a single beep are heard identically, there is no mechanism to automatically detect, by any node, that the beep from $x$ was dominant at node $y$, not even that there was a dominant beep at node $y$, regardless of the beeper. 
Such detection needs to be part of an algorithm -- in fact, our second algorithm discovers some dominant beeps during the computation (and decode information that this fact carries, using short coding sequences), but some with delay because it has to analyze some beeps from the past. 

The concept of ``dominating beeps'' divides the design of our main \probshort algorithm in two parts. First, we design an algorithm for a simpler OBLB problem, defined below (Definition~\ref{def:one-beep}), whose goal is to guarantee that each \heavy neighbor has a dominant beep in some slot. Second, the OBLB algorithm is enhanced with short coding sequences to assure that the receiver will eventually (during its beep sequence analysis) decode the dominating beep and the information it encodes.


\begin{definition}
\label{def:one-beep}
Given a constant \weight threshold $\sigma\geq 1$ and a WBN with set of \heavy links $E_\sigma$, the \defn{One-beep Local Broadcast (OBLB)} problem is solved in $\Gamma$ rounds if, for each \heavy link $(x,y)\in E_\sigma$, by the end of slot $\Gamma$ 
there has been a dominant beep from~node~$x$~at~node~$y$.
\end{definition}

\subsection{\mwd Metric Function of Algorithms}
\label{sec:algs}
We measure algorithmic performance in rounds as a function of $|E_{\sigma}|$ and the following characterization function of a WBN, called \defn{\mwdfulllong}, also called \defn{\mwdfull} (\defn{\mwd}) for short: 

\remove{
\begin{align*}
\boxed{
I_\sigma = \max_{E'\subseteq E_\sigma} \frac{1}{|E'|}\sum_{(x,y)\in E'} \sum_{(z,y)\in E} w(z,y)
\ .
}
\end{align*}
\mwd is the \mwdfull the receiver of an average link in any subset of $E_\sigma$ (links with \weight above threshold) can experience, from the transmitters of all other links.
}

\mgl{
\begin{align*}
\boxed{
I_\sigma = \max_{(x,y) \in E_\sigma} \sum_{(z,y)\in E:z\neq x} w(z,y)
\ .
}
\end{align*}
\mwd is the maximum aggregated weight of all other links incoming to a receiver of any link in $E_\sigma$ (i.e., the set of \heavy links).}
Increasing the parameter $\sigma$ we restrict the information to be learned (neighborhoods may shrink), but we may learn it faster as $I_\sigma$ is monotonically non-increasing with
growing $\sigma$. In Figure~\ref{fig:model} 
for example, setting $\sigma=1$ means that only $4$ links 
have to have dominant beeps.
Nevertheless, links with smaller \weight may still 
cause hearing a beep by the receiver, which has to be mitigated by the algorithm.

\section{Algorithmic Upper Bound}
\label{s:upbound}

We first present and analyze a solution to the simplified 
version of \probshort,
the One-beep Local Broadcast (OBLB) as in Def.~\ref{def:one-beep} (Sec.~\ref{sec:one-beep}). 
Next,
we generalize it to any 
\probshort
input (Sec.~\ref{sec:general-algorithm}). Extensions of our algorithms are given in Sec.~\ref{sec:extensions}.

\subsection{One-beep Local Broadcast}
\label{sec:one-beep}

Consider the following 
One-beep Local Broadcast algorithm:

\begin{center}
\boxed{
\textit{In each time slot, each node beeps with probability $p\leq 
\frac{1}{4I_\sigma}$ .}
}
\end{center}

The following analysis shows 
correctness and upper bound on the running time 
of this algorithm.

\begin{lemma}
\label{lemma:fraction}
Consider any subset of links $E'_\sigma\subseteq E_\sigma$, such that for each $(x,y)\in E'_\sigma$ node $x$ beeps with probability $p\leq 1/(4I_\sigma)$ 
\miguel{or listens with probability $1-p$.}
Then, the expected number of links 
$(x,y)\in E'_\sigma$ where a beep from $x$ is dominant at $y$
is at least $|E'_\sigma|(1-p)p/4$.
\end{lemma}

\mgl{
\begin{proof}
For a given local broadcast protocol, let $X_{x}(t)$ be a random variable indicating \miguel{($1$ or $0$)} whether node $x\in V$ beeps at time~$t$, and let $X_{x,y}(t)$ be a random variable indicating \miguel{($1$ or $0$)} whether a beep from $x$ was 
dominant at
$y$ at time $t$.
To prove the lemma, we need to show: $\mathbf{E}\left(\sum_{(x,y)\in E'_\sigma}X_{x,y}(t)\right) \geq |E'_\sigma|(1-p)p/4$.
We have that
\begin{align}
\mathbf{E}\Big(\sum_{(x,y)\in E'_\sigma}X_{x,y}(t)\Big) 
&= \sum_{(x,y)\in E'_\sigma} Pr(X_{x,y}(t)=1)\nonumber\\
&= \sum_{(x,y)\in E'_\sigma}
Pr(X_{x,y}(t)=1|X_x(t)=1\land X_y(t)=0)
\cdot p
(1-p)
\ .
\label{eq:lemma3}
\end{align}

Recall that, for any $(x,y)\in E$, the conditions 
for a beep from node $x$ to be dominant at node $y$
at time $t$ are:
    $X_y(t) = 0$,
    \miguel{$X_x(t)=1$,
    $w(x,y) \geq \sigma$,
    and 
    $\sum_{z\in \mm{V(t)}:z\neq x} w(z,y) < 1$,
    where $V(t)\subseteq V$ is the set of nodes beeping at time $t$.
From such conditions
we now bound the complementary event.
Defining the random variable $Z_{x,y}(t)=\sum_{z\in V(t):z\neq x} w(z,y)$ \dk{for any $(x,y)\in E_\sigma$}, we have the following:}
$$Pr\Big(X_{x,y}(t)=0|X_x(t)=1\land X_y(t)=0\Big) = 
\miguel{Pr\Big(Z_{x,y}(t) \geq 1\Big)} \ .$$

By Markov inequality,\footnote{Note that Markov inequality does not require independence of the random variables (see Thm.~3.1~in~\cite{MitzenmacherUpfal-book2005}).} the latter is 

$$Pr\Big(Z_{x,y}(t) \geq 1\Big)\leq \mathbf{E}(Z_{x,y}(t)) = p\sum_{z\in V:z\neq x} w(z,y)\ .$$

Upper bounding $\sum_{z\in V:z\neq x} w(z,y)\leq I_\sigma$ (by definition of $I_\sigma$) and  $p\leq 1/(4I_\sigma)$ (by the condition of the lemma), we get
$Pr\left(X_{x,y}(t)=0|X_x(t)=1\land X_y(t)=0\right) 
\leq 1/4$.
Putting them all together \dk{in Eq.~\ref{eq:lemma3}}, we have that
$
\mathbf{E}\left(\sum_{(x,y)\in E'_\sigma}X_{x,y}(t)\right) 
\geq
\sum_{(x,y)\in E'_\sigma} \left( p(1-p)/4 \right)
\geq 
|E'_\sigma|p(1-p)/4
$, as claimed.
\end{proof}
}

\begin{theorem}
\label{thm:One-beep-LB}
For $p\leq \frac{1}{4I_\sigma}$, One-beep Local Broadcast can be solved in~$O(I_\sigma \log |E_\sigma|)$ steps~whp.
\end{theorem}

\begin{proof}
Consider an execution of a One-beep Local Broadcast protocol. Let $E_{\sigma}^{(t)}\subseteq E_\sigma$ be the set of links with \weight at least $\sigma$ where 
a beep
was not 
dominant
before the beginning of time slot $t=1,2,\dots$. That is, $E_{\sigma}^{(1)} = E_\sigma$.
By Lemma~\ref{lemma:fraction}, for $p\leq 1/(4I_\sigma)$ we have that for any $t\geq 1$ it is
%
$\mathbf{E}\left(|E_\sigma^{t+1} |\ \big| \ |E_\sigma^t|=i\right)\leq (1-(1-p)p/4)i$.
%

Then, using conditional expectations, we get:
\remove{
\begin{eqnarray*}
\mathbf{E}\left(|E_\sigma^{(t+1)}|\right) 
&=& \sum_{i=0}^{|E_\sigma|} Pr(|E_\sigma^{(t)}|=i) \cdot \mathbf{E}\left(|E_\sigma^{t+1}| \ \big| \ |E_\sigma^t|=i\right)\\
&\leq&  \sum_{i=0}^{|E_\sigma|} Pr(|E_\sigma^{(t)}|=i) (1-(1-p)p/4)i
\\
&=& (1-(1-p)p/4) \mathbf{E}\left(|E_\sigma^{(t)}|\right)
\ = \ (1-(1-p)p/4)^t |E_\sigma|
\ .
\end{eqnarray*}
}
\begin{multline*}
\mathbf{E}\left(|E_\sigma^{(t+1)}|\right) 
= \sum_{i=0}^{|E_\sigma|} Pr(|E_\sigma^{(t)}|=i) \cdot \mathbf{E}\left(|E_\sigma^{t+1}| \ \big| \ |E_\sigma^t|=i\right) \leq \\
\ \ \le
\sum_{i=0}^{|E_\sigma|} Pr(|E_\sigma^{(t)}|=i) (1-(1-p)p/4)i
= (1-(1-p)p/4) \mathbf{E}\left(|E_\sigma^{(t)}|\right)
\ = \ (1-(1-p)p/4)^t |E_\sigma|
\ .
\end{multline*}
Notice for the latter inequality that transmitters do not need to switch off because in Lemma~\ref{lemma:fraction} the \weight is added over all neighbors, regardless of whether 
their beep was dominant already~or~not. 

On the other hand, by Markov inequality,
the probability that \mm{One-beep} Local Broadcast has not been solved at the beginning of time slot $t+1$ is
\begin{align*}
Pr\left(|E_\sigma^{(t+1)}|\geq 1\right)
&\leq \mathbf{E}\left(|E_\sigma^{(t+1)}|\right)
\leq (1-(1-p)p/4)^{t} |E_\sigma|
\ .
\end{align*}
Fixing $t=(4/((1-p)p))\ln (n|E_\sigma|)$, which is in $O(I_\sigma\log n)$, we have that 
\begin{align*}
Pr\left(|E_\sigma^{(t+1)}|\geq 1\right)
&
\leq \left(1-\frac{(1-p)p}{4}\right)^{\frac{4}{(1-p)p}\ln (n|E_\sigma|)} |E_\sigma|
\leq \frac{1}{e^{\ln(n|E_\sigma|)}} |E_\sigma|
\ = \ 1/n
\ .
\end{align*}
Thus, the claim of the theorem follows (for $t$ steps and probability at least $1-1/n$, which could be boosted to any whp by multiplying value $t$ by suitable constant).
\end{proof}

\subsection{\probfull}
\label{sec:general-algorithm}


In this section, we generalize the OBLB algorithm and analysis from Section~\ref{sec:one-beep} to \mm{\probfull} 
(Definition~\ref{def:LB}). It is run locally and autonomously by each node $v$. The input contains, apart of the identifier of node $v$ (given by a binary string of length $\log n$) and an input message: the number of nodes $n$, the \mwd $I_\sigma$, and the maximum size of a message $\cM$. 
In Section~\ref{sec:extensions} we show how to get rid of assumptions of known $I_\sigma$ and synchronization, and relaxing the requirement of knowing $n$ to only a polynomial upper bound on $n$.

\paragraph*{Algorithm's structure.}
The pseudo-code of the \mm{\probfull} algorithm is given in Algorithm~\ref{alg:local-broadcast}.
The algorithm partitions \rounds into consecutive windows of time, called \defn{super-\rounds}, each of length $2\log n + \cM$, where $\cM$ is the maximum length of a message sent. 
The number of super-\rounds, denoted by $\tau$, is computed as in the time analysis of the One-beep Local Broadcast algorithm in Theorem~\ref{thm:One-beep-LB}.
The algorithm also uses a random generator Random$(1/(4I_\sigma))$ that produces $1$ with probability $1/(4I_\sigma)$, and $0$ otherwise.

\begin{algorithm}[t!]
\caption{\probfull algorithm, pseudo-code for a node $v$} 
    \label{alg:local-broadcast}
    \SetKw{In}{Input}
    \SetKw{Out}{Output}
    \SetKw{Break}{break}
\In (upper bounds on) parameters $n,I_\sigma,\cM$\\
    $\tau \gets $ upper bound on time in proof of Thm.~\ref{thm:One-beep-LB}, computed locally based on given~$n,I_\sigma$\\
    $\alpha_v \gets $ beeping pattern computed according to the definition of $\alpha_v$ in the main text \\
    \For{$i=1,2, \dots , \tau$}{
        \If{Random$(1/(4I_\sigma))=1$}{node $v$ \underline{beeps} in the $2\log n+\cM$ \rounds of the super-\round $i$ \underline{according to pattern $\alpha_v$}
        }
        \Else{node $v$ stays silent and records the feedback in the $2\log n+\cM$ \rounds of the super-\round $i$ into string $\beta_{v,i}$
        }
    }
\For{$i=1,2, \dots , \tau$}{
    \If{$\beta_{v,i}$ is valid}{node $v$ locally (no beeping \round needed) \underline{decodes} $\beta_{v,i}$ into a pair~(ID,~message)}
    }
\Out decoded pairs (without redundancies), each containing node ID and message
\end{algorithm}

\paragraph*{Super-\round $i$.}
In the beginning of a super-\round $i$, each node $v$ performs local action as in \round $i$ of the One-beep Local Broadcast algorithm (see Section~\ref{sec:one-beep}). It uses the random generator. If the decision is to beep, then it 
\defn{beeps} in its $i$th super-\round\ 
\defn{according to the following beeping pattern, denoted $\alpha_v$}:
\begin{itemize}
\item 
if there is bit 1 in position $j$ of ID of node $v$, it beeps and stay silent  in \rounds $2j-1$ and $2j$ of its super-\round, resp.
\item 
if there is bit 0 in position $j$ of ID of node $v$, it stay silent and beeps in \rounds $2j-1$ and $2j$ of its super-\round, resp.
\item 
in \rounds $2\log n+1,\ldots, 2\log n +\cM$ of its super-\round, node $v$ beeps according to the bit representation of its message (i.e., it beeps for bit 1 and stays silent for bit~0); if there are less than $\cM$ bits in the message, it piggybacks the message with the 0's until the end of the super-\round.
\end{itemize}
If, however, the decision is not to beep, node $v$ stays silent during the whole super-\round $i$ and records the received sequence of beep/silence into a binary string $\beta_{v,i}$.

We call a super-\round $i$ \defn{valid} for node $v$ if $v$ has been listening during the super-\round and its recorded feedback $\beta_{v,i}$ has the following property: the number of 1's in the first $2\log n$ positions of $\beta_{v,i}$ is exactly $\log n$. 

\paragraph*{Final decoding.}
At the end of Algorithm~\ref{alg:local-broadcast}, node $v$ \textbf{\emph{decodes}} each $\beta_{v,i}$ of a valid super-\round~$i$:
\begin{itemize}
\item 
if there are bits 10 at positions $2j-1,2j$, for $j\le \log n$, it decodes bit 1 and puts it at position $j$ of the decoded~ID;
\item 
if there are bits 01 at positions $2j-1,2j$, for $j\le \log n$, it decodes bit 0 and puts it at position $j$ of the decoded~ID;

\item 
in \round $2\log n + j$, for $j\le \cM$, node $v$ decodes 1 and puts it in position $j$ of the decoded message if it heard a beep, and decodes 0 if there was a silence.
\end{itemize}
Then, node $v$ puts pairs, containing decoded ID and message, to its output, removing redundancies.



\paragraph*{Algorithm analysis.}

\begin{lemma}
\label{lem:algo-progress}
If there is only one \heavy neighbor of a node $v$ choosing to beep in the considered super-\round $i$, and the aggregated \weight on $v$ from other neighbors in this super-\round is smaller than $1$, then node $v$ correctly decodes the other node's ID and message from sequence $\beta_{v,i}$ at the end of the algorithm.
\end{lemma}

\remove{
\begin{proof}[Proof of Lemma~\ref{lem:algo-progress}]
The crucial observation is that if there is only one \heavy neighbor of $v$ that chooses to beep in the super-\round $i$, say node $v^*$, and the aggregated \weight from other neighbors of $v$ is smaller than $1$, then the operations of beeping according to pattern $\alpha_{v^*}$ is directly received by $v$, in the sense that $\beta_{v,i}=\alpha_{v^*}$. In such case, observe that decoding is the actual reversed operation to encoding the pattern (i.e., reversed to creating pattern $\alpha_{v^*}$), and in this sense can uniquely decode the ID of $v^*$ (based on the first $2\log n$ bits of sequence $\beta_{v,i}$) and the message of $v^*$ (based on the remaining $\cM$ bits of the sequence).
\end{proof}
}
\begin{proof}
The crucial observation is that if there is only one \heavy neighbor of $v$ that chooses to beep in the super-\round $i$, say node $v^*$, and the aggregated \weight from other neighbors of $v$ is smaller than $1$, then the operations of beeping according to pattern $\alpha_{v^*}$ is directly received by $v$, in the sense that $\beta_{v,i}=\alpha_{v^*}$. In such case, observe that decoding is the actual reversed operation to encoding the pattern (i.e., reversed to creating pattern $\alpha_{v^*}$), and in this sense can uniquely decode the ID of $v^*$ (based on the first $2\log n$ bits of sequence $\beta_{v,i}$) and the message of $v^*$ (based on the remaining $\cM$ bits of the sequence).
\end{proof}

\begin{lemma}
\label{lem:algo-correct}
If two or more \heavy neighbors of a node $v$ beep in the beginning of a super-\round $i$, then node $v$ ignores the information received in this super-\round in sequence $\beta_{v,i}$.
\end{lemma}

\remove{
\begin{proof}[Proof of Lemma~\ref{lem:algo-correct}]
It follows from the fact that $v$ can recognize that at least two of its \heavy neighbors beeped, or that the aggregated \weight from other neighbors is at least $1$, by analyzing the first $2\log n$ \rounds of this super-\round. More precisely, among these bits, there would be more than $\log n$ beeps in these \rounds, while in the case of one \heavy neighbor and aggregated \weight of others smaller than $1$ there should be exactly $\log n$ beeps in the received sequence~$\beta_{v,i}$.
\end{proof}
}
\begin{proof}
It follows from the fact that $v$ can recognize that at least two of its \heavy neighbors beeped, or that the aggregated \weight from other neighbors is at least $1$, by analyzing the first $2\log n$ \rounds of this super-\round. More precisely, among these bits, there would be more than $\log n$ beeps in these \rounds, while in the case of one \heavy neighbor and aggregated \weight of others smaller than $1$ there should be exactly $\log n$ beeps in the received sequence~$\beta_{v,i}$.
\end{proof}

\begin{theorem}
\label{thm:algorithm-LB}
For $0<p\leq 1/(4I_\sigma)$, \mm{\probfull} can be solved in $O(I_\sigma \log n \cdot (\log n+\cM))$ time slots whp.
\end{theorem}

\begin{proof}
We combine Theorem~\ref{thm:One-beep-LB} with Lemma~\ref{lem:algo-progress} to conclude that for any node $v$ and any of its \heavy neighbors $v^*$, there is a super-\round $i$ when $v^*$ is the only \heavy neighbor of $v$ that beeps, and the aggregated \weight of other neighbors of $v$ is smaller than the channel threshold $1$ (hence, the beeps of $v^*$ within the super-\round are heard, as it is \heavy, while the slots when $v^*$ does not beep are silent at $v$ due to the threshold channel condition). Thus, the ID and message of $v^*$ are correctly decoded by 
$v$ at the end of the algorithm.

Lemma~\ref{lem:algo-correct} implies that node $v$ correctly ignores information that could come from more than one \heavy neighbor or be ``too noisy'' by aggregated \weight of others -- node $v$ correctly recognize it as non-valid.

The number of \rounds is $\tau \cdot (2\log n + \cM)$. By Theorem~\ref{thm:One-beep-LB}, which analysis is used to compute $\tau$, and by asymptotic $|E_\sigma|=\Theta(n)$, the total number of time slots is $O(I_\sigma \log n \cdot (\log n+\cM))$.
\end{proof}




\subsection{Extensions and Dropping Off Some Limitations}
\label{sec:extensions}

Our main algorithm in Section~\ref{sec:general-algorithm} can be extended in the following ways.

\paragraph*{Arbitrary (unknown) $I_\sigma$ and only polynomial upper bound on $n$.}

First, observe that in the algorithm we only need a polynomial upper bound on $n$,
as $n$ occurs under logarithm in the formula on \mwd in Theorem~\ref{thm:algorithm-LB}. 

The impact of an estimation of the actual parameter $I_\sigma$
 could be, hypothetically, more significant, as the number of learning \rounds depends linearly on $I_\sigma$
 (see Thm.~\ref{thm:algorithm-LB} and later Thm.~\ref{thm:lower}).
Our algorithm could be transformed into one without the need of knowing $I_\sigma$
in advance, by using doubling estimates~of~$I_\sigma$ (i.e., we keep running our algorithm for parameters $I_\sigma$ being subsequent powers~of~$2$) until 
\miguel{no beep is dominant.}
The main property that makes it feasible is described in Lemma~\ref{lem:algo-correct} -- receiving nodes can recognize whether the current communication attempt comes from one agent (which is desired) or more than one (in which case, the received sequence of beeps can be ignored).
Additionally, acknowledgments of successful transmissions have to be implemented: by repeating the same algorithm, but this time with messages being acknowledgments of previously received messages. Note that this method increases the number of steps only by a constant factor.



\paragraph*{De-synchronized systems.}
 Due to a relatively simple structure and various checking mechanisms in our algorithm, it 
 \mmr{can} be extended to de-synchronized settings (with clock shifts) by using the \mmr{known} synchronizing methods for beeping networks, see e.g., the recent work~\cite{DeMarcoK25}. These methods require only a constant overhead, and transform algorithms with structures similar to ours (i.e., main loop iterating a preamble code and message) from fully synchronized into de-synchronized~solutions.

\paragraph*{Dynamic weights and noise.}


In our model, physical conditions, such as noise, are embedded in link weights. The question of fluctuating physical conditions then boils down to how dynamic weights impact the complexity of \probshort. Under arbitrary dynamicity, \probshort
is not computable in 
a WBN: if neighborhoods may change in every \round, the computation must have taken at most one \round. This observation yields that \probshort could be defined for restricted dynamic settings, for instance, if weights change only every so many \rounds, or changes are driven by a stochastic process with bounded deviation. Our algorithms can be applied to such settings as well, as long as weights are locally static (or nearly static, in the sense that at any time point, their \mwd does not exceed a given upper bound) during the execution of the algorithm.
This is because Lemma~\ref{lemma:fraction} still holds for any single \round of such dynamic model, while the way it is used in the final proofs of Theorem~\ref{thm:One-beep-LB} and Theorem~\ref{thm:algorithm-LB} is independent on specific values of weights but depends only the upper bound on \mwd.  
More dynamic scenarios, e.g., when \mwd of current weights may oscillate, are left as a challenging open problem. Note that our lower bound (for static weights) automatically extends to dynamic scenarios, and thus gives a point of reference to the future study of more dynamic scenarios.

\section{Lower Bound}
\label{sec:lower-bound}

In this section, we prove that any \probfull algorithm takes at least $\Omega(I_\sigma)$ beeping \rounds in some WBN.
In fact, even a simplified one dominant beep per \heavy neighbor version of the problem requires such number of beeping \rounds in some networks.

\mm{In the analysis of this section we will show an adversarial WBN that requires the claimed running time. The network topology of such WBN is a complex bipartite graph (non-existing links correspond to weight $0$ between the end nodes).}
%
We start with a technical lemma 
followed by the main~theorem.

\begin{lemma}
\label{lemma:agginf}
Consider a WBN with a complete bipartite topology graph $(B,L,E)$, where $B$ is a set of beepers, $L$ a set of listeners, $E=\{(b,\ell)|b\in B \land \ell\in L \}$, and $|B|=|L|$.
The network is embedded in a metric space with distance function $d:B\times L \to \mathbb{R}^+$. 
The \weight function is $w(b,\ell)=d_{\min}/d(b,\ell)$ for all $(b,\ell)\in E$, where $d_{\min}$ is the smallest length among all links in $E$. 
Consider some link $(b',\ell')\in E$ and a set of links $\cE\subset E$ such that 
$(b',\ell')\notin \cE$. 
Then, if it is possible 
to have a dominant beep
in each link $(b,\ell)\in \cE$ in one time slot, it is $\sum_{(b,\ell)\in \cE} w(b,\ell')\in O(1).$
\end{lemma}



\begin{figure}[t!]
  \centering
  \includegraphics[width=.6\linewidth]{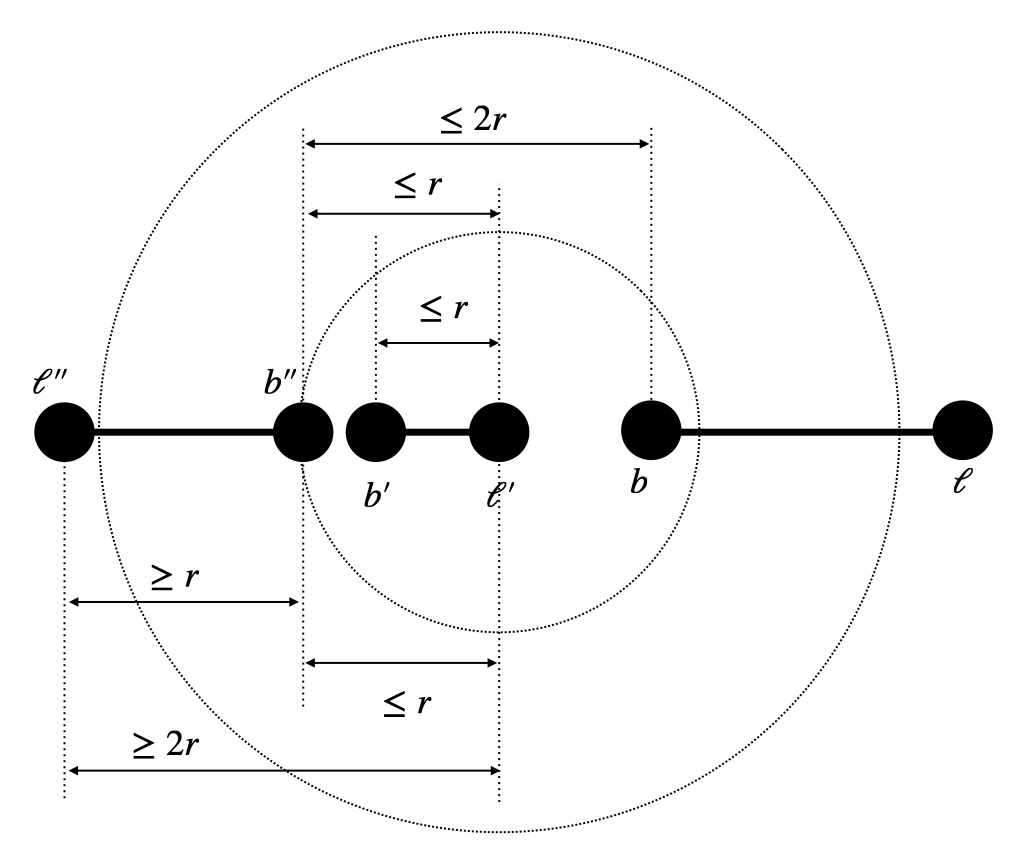}
  \caption{  Distance bound $d(b,\ell'')\leq 3d(b'',\ell'')$, for some link $(b'',\ell'')\in \cE_{close}$, which yields $\sum_{(b,\ell)\in\cE_{close}}w(b,\ell')<4$ using that a beep is 
\miguel{dominant in}
link $(b'',\ell'')$.
  }
  \label{fig:folding}
\end{figure}

\begin{proof}
To prove the claim we split the set of links $\cE$ into \defn{distant} and \defn{close} with respect to listener $\ell'$.
For the set $\cE_{close}$ of close links, we upper bound the \weight of their beepers on $\ell'$ simply by $|\cE_{close}|$ 
(since each link \miguel{$(b,\ell)\in E$} has \weight 
\miguel{$w(b,\ell)=d_{\min}/d(b,\ell)\leq 1$}),
and we bound $|\cE_{close}|$ by a constant using a geometric argument and the fact that a beep
\miguel{is dominant in each}
 link in $\cE_{close}$. 
For the set $\cE_{dist}$ of distant links, we bound the \weight of their beepers on $\ell'$, by their \weight on a listener at shortest distance from $\ell'$. We bound the latter by a constant using a geometric argument and the fact that a beep
\miguel{is dominant in each}
 link in $\cE$.
The details~follow.

Let $2r$ be the shortest distance from any listener in $\cE$ to listener $\ell'$.
Let $\cE_{close}=\{(b,\ell)\in \cE| d(b,\ell')\leq r\}$, that is the set of links $(b,\ell) \in \cE$ whose beeper $b$ is within distance $r$ from listener $\ell'$. 

We bound first the \weight of close beepers on $\ell'$ as 
\begin{align}
\sum_{(b,\ell)\in~\cE_{close}} w(b,\ell') 
= \sum_{(b,\ell)\in~\cE_{close}} \frac{d_{\min}}{d(b,\ell')} 
\leq |\cE_{close}|
\ .
\label{eq:aggonv}
\end{align}
We now bound $|\cE_{close}|$ as follows. 
Fix a link $(b'',\ell'')\in~\cE_{close}$.
The aggregated \weight of the other links in $\cE_{close}$ on listener $\ell''$ is
\begin{align}
\sum_{(b,\ell)\in~\cE_{close}\setminus\{(b'',\ell'')\}} w(b,\ell'') &= \sum_{(b,\ell)\in ~\cE_{close}\setminus\{(b'',\ell'')\}} \frac{d_{\min}}{d(b,\ell'')}
\ . 
\label{eq:agg}
\end{align}
The equality in the latter is due to the graph being complete bipartite. 
The boundary case of the following geometric argument is illustrated in Figure~\ref{fig:folding}.

\begin{figure}[t!]
  \centering
  \includegraphics[width=.6\linewidth]{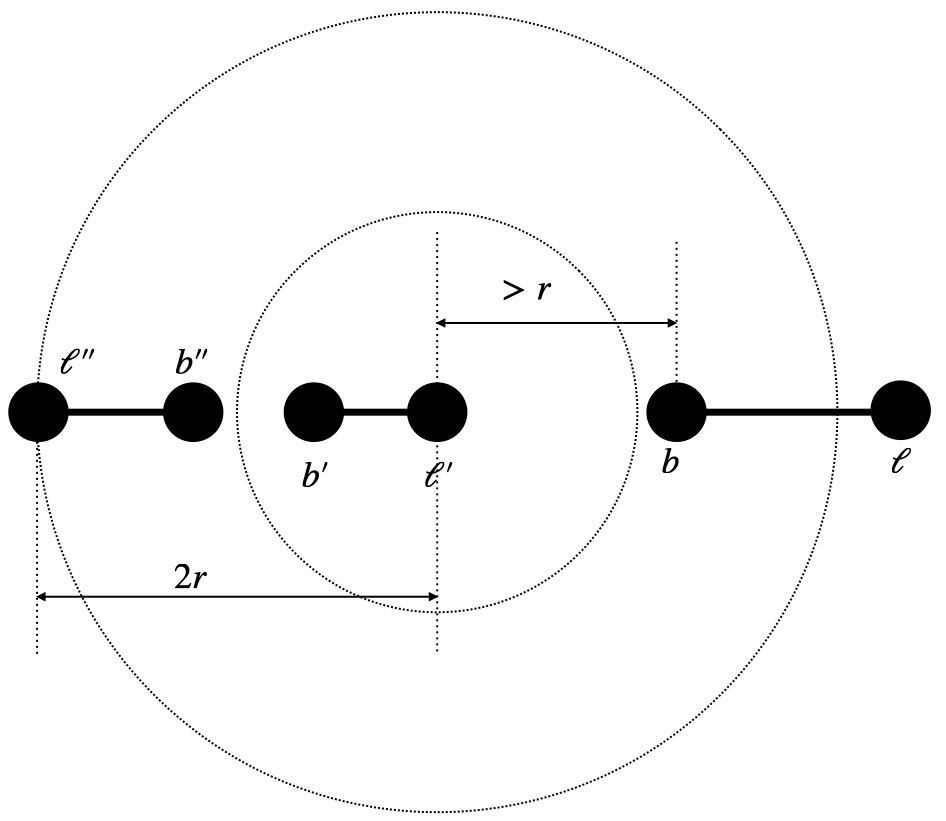}
  \caption{Distance bound $d(b,\ell'')\leq 3d(b,\ell')$, for some link $(b'',\ell'')\in \cE$ whose listener is at shortest distance $2r$ of $\ell'$, which yields $\sum_{(b,\ell)\in~\cE_{dist}} w(b,\ell')< 6$ using that a beep is 
\miguel{dominant in}
link $(b'',\ell'')$.
}
  \label{fig:folding2}
\end{figure}

Given that $b''$ is within distance at most $r$ of $\ell'$ (because $(b'',\ell'')\in~\cE_{close}$), and $\ell''$ is within distance at least $2r$ from $\ell'$ (because $2r$ is the shortest distance from any listener in $\cE$ to $\ell'$), the length of $(b'',\ell'')$ is at least $r$  (i.e. $r\leq d(b'',\ell'')$).
Also, given that $b$ and $b''$ are both within distance $r$ of $\ell'$ (because $(b,\ell)\in~\cE_{close}$ and $(b'',\ell'') \in~\cE_{close}$), we have that $b$ is within distance $2r$ of $b''$ (i.e. $d(b,b'')\leq 2r$).
Combining and using the triangle inequality, we have that 
\[
d(b,\ell'')\leq d(b,b'')+d(b'',\ell'')\leq 2r+d(b'',\ell'')\leq 3d(b'',\ell'') \ .
\]
Replacing in Equation~\ref{eq:agg},
\begin{align}
\sum_{(b,\ell)\in~\cE_{close}\setminus\{(b'',\ell'')\}} w(b,\ell'') 
&\geq \frac{|\cE_{close}|-1}{3}~\frac{d_{\min}}{d(b'',\ell'')} 
= \frac{|\cE_{close}|-1}{3}~w(b'',\ell'')
\ . 
\label{eq:aggonE}
\end{align}
On the other hand, given that 
a beep is 
\miguel{dominant in}
each link in $\cE$, in particular 
a beep is 
\miguel{dominant in}
link $(b'',\ell'')$. Therefore, using the conditions for 
\miguel{dominance}
we know that
\[
\sum_{(b,\ell)\in~\cE_{close}\setminus\{(b'',\ell'')\}} w(b,\ell'') < w(b'',\ell'') \ .
\]
Replacing the latter in Equation~\ref{eq:aggonE}, it is
\mam{$(|\cE_{close}|-1)w(b'',\ell'')/3 < w(b'',\ell'')$, that is,}
$|\cE_{close}| < 4$. Replacing in Equation~\ref{eq:aggonv}, we get
\begin{align}
\sum_{(b,\ell)\in~\cE_{close}} w(b,\ell') < 4.\label{eq:boundin}
\end{align}
We now bound the \weight of beepers in $\cE_{dist}=\cE\setminus \cE_{close}$ on listener $\ell'$. We have that
\begin{align}
    \sum_{(b,\ell)\in~\cE_{dist}} w(b,\ell') = \sum_{(b,\ell)\in~\cE_{dist}} \frac{d_{\min}}{d(b,\ell')}.\label{eq:outside}
\end{align}
The equality in the latter is due to the graph being complete \dk{bipartite}.
The boundary case of the following geometric argument is illustrated in Figure~\ref{fig:folding2}.

Let $(b'',\ell'')\in \cE$ be a link such that $\ell''$ is a listener at shortest distance $2r$ from listener $\ell'$. 
By triangle inequality, for any $(b,\ell)\in~\cE_{dist}$, it is $d(b,\ell'')\leq d(b,\ell')+d(\ell',\ell'') = d(b,\ell')+2r$.
Given that, for any $(b,\ell)\in~\cE_{dist}$, it is $r\leq d(b,\ell')$, we have that $d(b,\ell'')\leq 3d(b,\ell')$. 
Replacing $d(b,\ell')\geq d(b,\ell'')/3$ in Equation~\ref{eq:outside}, we get
\begin{eqnarray*}
    \sum_{(b,\ell)\in~\cE_{dist}} w(b,\ell') 
    &\leq& \sum_{(b,\ell)\in~\cE_{dist}} \frac{3d_{\min}}{d(b,\ell'')}
    \ \ = \ \ 3 \sum_{(b,\ell)\in~\cE_{dist}} w(b,\ell'')
    \\
    &\leq& 3 \left(w(b'',\ell'')+\sum_{(b,\ell)\in~\cE_{dist}\setminus\{(b'',\ell'')\}} w(b,\ell'')\right)
    \ .
\end{eqnarray*}
%
Given that 
a beep is dominant in
each link in $\cE$, in particular it is 
dominant in
$(b'',\ell'')$. Hence, by the conditions to have a 
dominant beep in
link $(b'',\ell'')$, it is $\sum_{(b,\ell)\in~\cE_{dist}\setminus\{(b'',\ell'')\}} w(b,\ell'')
< 1$. 
On the other hand, it is $w(b'',\ell'')\leq 1$ because 
for the input WBN
used in this proof all 
$w(\cdot,\cdot)$ are at most $1$. Replacing in the latter we get 
\begin{align}
\sum_{(b,\ell)\in~\cE_{dist}} w(b,\ell')< 
6
\ .
\label{eq:boundout}
\end{align}
%
Combining the bounds in 
Eqs.~\ref{eq:boundin} and~\ref{eq:boundout} the claim follows.
\end{proof}


\begin{figure}[t!]
\centering
\includegraphics[scale=0.29]{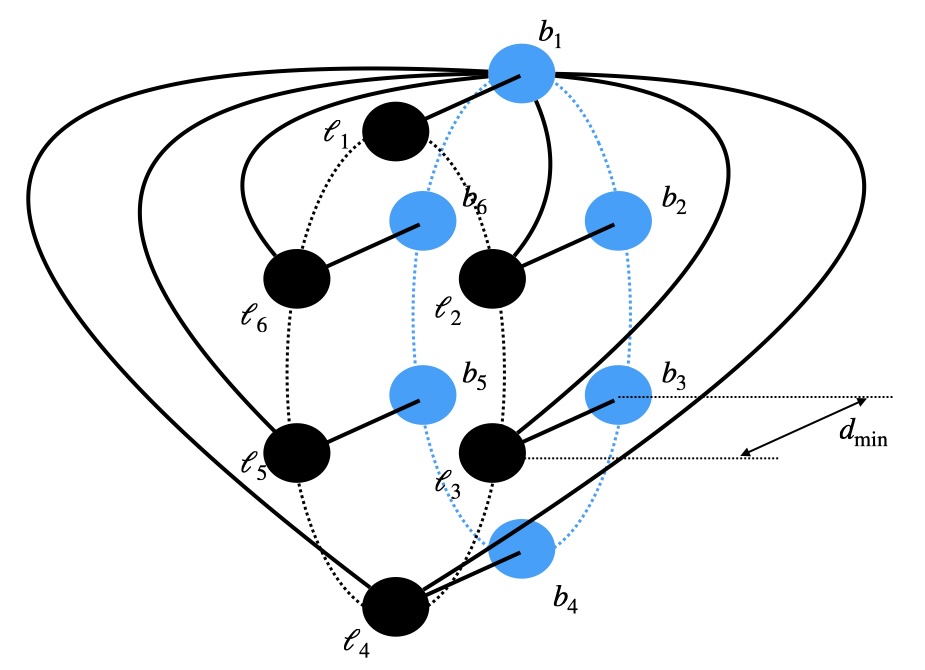}
\caption{Example of adversarial network for $n=6$ for the proof of Theorem~\ref{thm:lower}.}
\label{fig:wheels}
\end{figure}

\begin{theorem}
\label{thm:lower}
There exists a \mm{WBN} and \weight function such that, to solve the \mm{\probfull} problem, $\Omega(I_\sigma)$ time slots are required with probability $1$, where $I_\sigma$ is the 
\mwd.
\end{theorem}


\begin{proof}
We prove the claim showing an adversarial \mm{WBN} where even One-beep Local Broadcast in a subset of links requires the claimed running time. 
Specifically, consider a complete bipartite \mm{WBN} $(B,L,E)$ as required by Lemma~\ref{lemma:agginf}.
Let the nodes be labeled as $B=\{b_1,b_2,\dots,b_n\}$ and $L=\{\ell_1,\ell_2,\dots,\ell_n\}$, for some $n>1$.
Nodes are placed in space so that,
for some $d_{\min}>0$ 
and for all $i,j\in [n]$,
if $i\neq j$ it is $d(b_i,\ell_j)>d_{\min}$, 
or otherwise it is $d(b_i,\ell_j)=d_{\min}$,
and every beeper (resp. listener) has the same distribution of \weight among its outgoing (resp. incoming) links
(see a deployment example in Figure~\ref{fig:wheels}).

\remove{
\begin{figure}[htbp]
\centering
\includegraphics[scale=0.2]{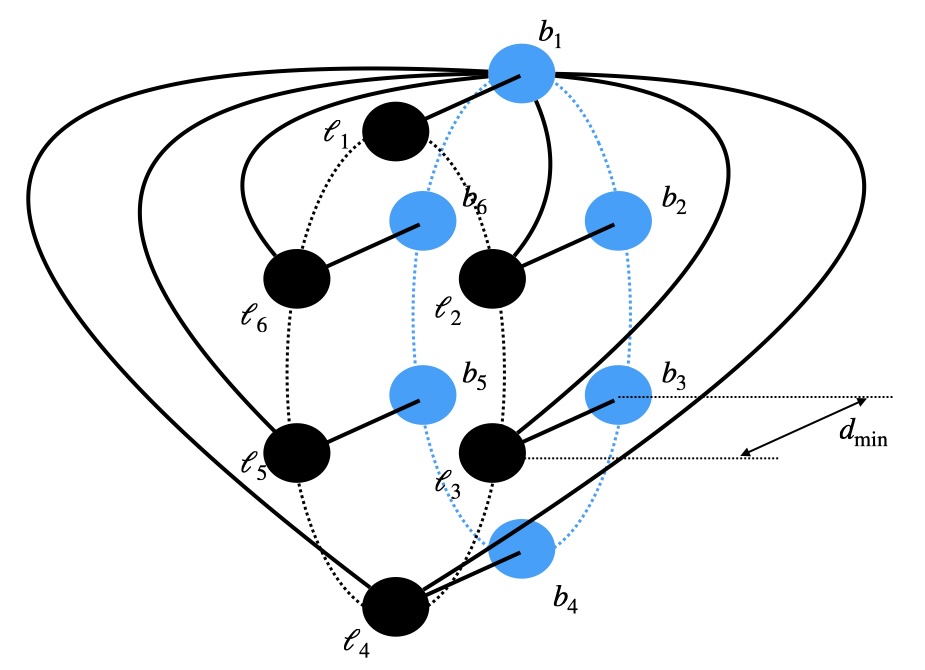}
\caption{Example of adversarial network for $n=6$.}
\label{fig:wheels}
\end{figure}
}

Under the above definition the \weight threshold is $1$ and the set of \heavy links is $E_1=\{(b_i,\ell_i)|i\in [n]\}$.
Recall 
that
%
\mgl{$I_\sigma = \max_{(x,y) \in E_1} \sum_{(z,y)\in E:z\neq x} w(z,y)$.}
%
%
\mgl{
Given that all listeners have the same \weight distribution among their incoming links, any link in $E_1$ maximizes $I_\sigma$, for instance, $(b_1,\ell_1)$. Then,
}
\remove{
\begin{align*}
I_\sigma 
&= 
\frac{1}{n}\sum_{(b,\ell)\in E_1} \sum_{(b',\ell)\in E} w(b',\ell)
= \frac{1}{n}~n \sum_{(b,\ell)\in E_1} w(b,\ell_1)
= 1 + \sum_{(b,\ell)\in E_1\setminus\{(b_1,\ell_1)\}} w(b,\ell_1)
\ .
\end{align*}
}
\mgl{
\begin{align*}
I_\sigma 
&= 
\sum_{(b,\ell)\in E_1\setminus \{(b_1,\ell_1)\}} w(b,\ell_1)
\ .
\end{align*}
}
%
%
Let $T(E_1)$ be the number of time slots needed to solve One-beep Local Broadcast on $E_1$ by some protocol~$\cP$.  
Applying Lemma~\ref{lemma:agginf} to each time slot of the execution of $\cP$, we have that for link $(b_1,\ell_1)\in E$ it is 
$
\sum_{(b,\ell)\in E_1\setminus\{(b_1,\ell_1)\}} w(b,\ell_1)\in O(T(E_1))
$.
The latter holds regardless of whether $\cP$ uses randomization or not. Thus, replacing, we get
$I_\sigma\in O(T(E_1))$ \mam{with probability~$1$,} and the claim follows.
\end{proof}

\section{General Emulation}
\label{sec:general-emulation}

Suppose we are given a problem specification $\cP$ and its solution $\cA$, both in the classic (unweighted) beeping networks model. We can emulate $\cA$, 
executed on
an (unweighted) beeping network $G'$, in any weighted beeping network $G\mgl{=(V,E)}$ satisfying $\overline{G}_\sigma = G'$, where $\overline{G}_\sigma$ is the unweighted version of $G_\sigma\mgl{=(V,E_\sigma)}$.
By ``emulation'' we mean that outputs of the corresponding nodes in both executions are the same. We call the emulating algorithm $\cA^*$, and its pseudo-code is in \emph{Algorithm~\ref{alg:emulator}} \mgl{(all line numbers in this section refer to the pseudocode in Algorithm~\ref{alg:emulator})}; it works on any WBN and for any $\sigma\ge 1$, with an integer parameter $c\ge 2$ that linearly scales the number of steps but exponentially reduces the error probability. More~details~follow.


\smallskip

\parhead{Histories, actions and outcomes}
A history vector $\cH_{v,\cA}$, indexed by steps and with values in $\{beep,0,1\}$, is implicitly built locally by a node $v$ running an algorithm $\cA$ on an underlying (unweighted) beeping network $G'$: for every integer index $t\ge 1$ for which it is defined, where $beep$ denotes that $v$ beeped in step $t$ while $1,0$ mean that $v$ was listening and heard a beep or no beep in step $t$, respectively. 
Such vector uniquely characterizes an execution of $\cA$ on the BN $G'$ (assuming random bits are already drawn), because the outcome of each step of node $v$ could be only one of the three cases described above.
By $\cA(\cH_{v,\cA}[0,\ldots,t])$ we denote the action of algorithm $\cA$ on the (partial) history vector $\cH_{v,\cA}[0,\ldots,t]$ at node $v$ after $t$ rounds, which could be either $beep$ or \mgl{$listen$} or $stop$ or not defined (which occurs continuously after the first action $stop$).\footnote{%
\dk{This may be further relaxed or extended in case of dynamic algorithms or a sleeping/wake-up models.}}
For the whole history vector $\cH_{v,\cA}$ at node $v$, the final output of the algorithm $\cA$ at node $v$ is denoted $Out(\cA,\cH_{v,\cA})$.\footnote{%
The analogous notation $\cH_{v,\cA^*}, \cA^*(\cH_{v,\cA^*}[0,\ldots,t]), Out(\cA^*,\cH_{v,\cA^*})$ could also be used for any algorithm $\cA^*$ that runs on a weighted beeping network $G$. However, in this paper we do not use this notation in the context of weighted beeping networks.
}

\smallskip

\parhead{Quasi-histories} 
The algorithm $\cA^*$, when run at a node $v$ in a weighted beeping network~$G$, incrementally builds and uses a quasi-history vector $\cH^*_{v,\cA^*}[t]\in \{n\}\times \{beep,0,1\}^* = \{n\}\times \bigcup_{i\ge 1} \{beep,0,1\}^i$, 
which is appended every \dk{$2c\log n$} \rounds of $\cA^*$ \footnote{%
$c\ge 2$ is an integer input parameter and, for the sake of avoiding rounding, we assume that $\log n$~is~an~integer. \dk{It could also suffice if the algorithm gets a (polynomial) upper bound on $n$ as part~of~the~input.}}
(see lines~\ref{alg-item:emulation-update} - \ref{line:set-beep} )
 and, as we will show later, resembles the actual history vector $\cH_{v,\cA}$ of the emulated algorithm $\cA$ as it would be produced when run at node $v$ of the unweighted version of graph~$G_\sigma$, denoted $\overline{G}_\sigma$.
\dk{The inputs of the emulated algorithm $\cA$ and the emulator $\cA^*$ are the same, and since the inputs could be modeled as the beginning values of the history vectors, the emulator $\cA^*$ sets it in lines~\ref{line:initial} and \ref{line:index-0}. 
}

\smallskip

\parhead{Emulating a single step -- distinguishing a beep from a no-beep}
In order to build quasi-histories, nodes have to distinguish a beep from a no-beep in a single emulated step of algorithm $\cA$.
\dk{
\mgl{For that purpose, each node $v$ running $\cA^*$ builds a sequence $\alpha_v$ of bits, first choosing bits locally and independently at random, and then encoding each of the bits obtained into two bits, see line~\ref{line:alpha}.} The whole sequence has length $2c\log n$, where $c$ is a constant parameter controlling the probability of success \mgl{of the emulation}. Each node uses its sequence $\alpha_v$ to determine when to beep, provided its quasi-history indicates beeping, in a ``for'' loop (see line~\ref{line:for-loop}). After that, listening nodes 
check if they received a majority of beeps or not -- if yes, they record a beep in their quasi history (reflecting that the node would likely receive a beep in the corresponding execution of $\cA$ on $\overline{G}_\sigma$), and recording no-beep otherwise. 
We will show in the analysis that this procedure on such sequences $\alpha_v$ guarantee the desired distinction between beep and no-beep as in the original algorithm $\cA$, while the overhead (i.e., the cost of distinguishing in a single original step) is $2c\log n$ in algorithm $\cA^*$.}


\begin{algorithm}[t!]

\caption{Emulator $\cA^*$ of a classic BN algorithm $\cA$, for each node $v$.
\label{alg:emulator}}
\DontPrintSemicolon
\SetKw{In}{Input}
\SetKw{Out}{Output}
\In 
constant $c\ge 2$;
~initial history $\cH_{v,\cA}[0]$;
~BN Algorithm given as function $\cA: \{beep,0,1\}^* \to \{beep,listen,stop\}$
\label{line:initial}
\;
$index\gets 0$, \ $\cH^*_{v,\cA^*}[0]\gets \cH_{v,\cA}[0]$
\label{line:index-0}
\;
\While{$\cA(\cH^*_{v,\cA^*})\ne stop$\label{line:loop-while}}{
$\alpha_v \gets$ node $v$ locally chooses a random binary sequence of length $c\log n$, such that each position is $1$ independently with probability $\frac{1}{6I_\sigma}$, and then (going through the sequence) it substitutes each $1$ by pair $10$ and each $0$ by pair $01$; thus,~$|\alpha_v|= 2c\log n$\label{line:alpha}\;
    $index \gets index + 1$, $\beta_v[1\dots |\alpha_v|] \gets$ empty vector\label{line:next-t}\;
    \For {each $i=1,\ldots,|\alpha_v|$\label{line:for-loop}}{
        \lIf{$\cA(\cH^*_{v,\cA^*}) = beep$ and $\alpha_v[i]=1$} {
        node $v$ beeps
        and sets $\beta_v[i]\gets 1$\label{line:beep-and-record}}
        \Else{ 
        node $v$ listens and
        \leIf{a beep is heard}{$\beta_v[i]\gets 1$}{$\beta_v[i]\gets 0$}
        }
    }
    \If{$\cA(\cH^*_{v,\cA^*})= listen$\label{alg-item:emulation-update}}{
        \lIf{at least half of the positions in $\beta_v$ are $1$'s}{$\cH^*_{v,\cA^*}[index]\gets 1$}
        \lElse{$\cH^*_{v,\cA^*}[index]\gets 0$}
    } 
    \lElse{$\cH^*_{v,\cA^*}[index]\gets \cH_{v,\cA}[index]$\label{line:set-beep}}
}
\Out $Out(\cA,\cH^*_{v,\cA^*})$ and $\cH^*_{v,\cA^*}$\;
\end{algorithm}

\remove{

\begin{enumerate}
\item[0.]
$t\gets 0$, \ $\cH^*_{v,\cA^*}[0]\gets \cH_{v,\cA}[0]$
\item 
$\alpha_v$ is a random binary sequence of length ???, selected locally by each node $v$, in which each one was substituted by a pair $10$ and each zero by the pair $01$.
\item 
{\bf While} $\cA(\cH^*_{v,\cA^*})\ne stop$ :
\label{line:loop-while}
\begin{enumerate}
\item 
$t\gets t+1$, $\beta_v$ is (re)set to empty vector
\label{line:next-t}
\item 
\label{line:for-loop}
{\bf For each} $i=1,\ldots,|\alpha_v|$ :
\begin{enumerate}
\item 
\label{line:beep-and-record}
{\bf If} $\cA(\cH^*_{v,\cA^*})= beep$ 
and $\alpha_v[i]=1$ {\bf then} $beep$ and set $\beta_v[i]\gets 1$
\\
{\bf Else} stay idle and record heard beep by $\beta_v[i]\gets 1$ and no beep by $\beta_v[i]\gets 0$
\end{enumerate}
\item 
\label{alg-item:emulation-update}
{\bf If} $\cA(\cH^*_{v,\cA^*}) = idle$ {\bf then} 
\begin{enumerate}
\item 
{\bf If} at least half of the positions in $\beta_v$ are $1$'s {\bf then}  
$\cH^*_{v,\cA^*}[t]\gets 1$ 
\\
{\bf Else} $\cH^*_{\cA^*}[t]\gets 0$
\end{enumerate}
{\bf Else} $\cH^*_{v,\cA^*}[t]\gets \cH_\cA(t)$
\end{enumerate}
\item[OUTPUT] $Out(\cA,\cH^*_{v,\cA^*})$ and $\cH^*_{v,\cA^*}$
\end{enumerate}

}

\subsection{Analysis}


The following theorem shows correctness of the emulation algorithm and its performance.
\mgl{Recall that we use $G'$ to denote a BN graph, $G=(V,E)$ to denote a WBN graph, $G_\sigma=(V,E_\sigma)$ to denote the subgraph of $G$ where $E_\sigma\subseteq E$ is the set of links with weight at least $\sigma$, and $\overline{G}_\sigma$ to denote the unweighted version of $G_\sigma$.}

\begin{theorem}
\label{thm:emulation}
Suppose we are given a specification of problem $\cP$ and an algorithm $\cA$ that 
\begin{itemize}
\item
correctly satisfies this specification in any (classic) BN $G'$ with no isolated nodes, and 
\item
terminates in $\uptau_\cA(n)=O(\poly(n))$ steps, 
\end{itemize}
with probability at least $1-n^{-c'}$, where $n$ is the number of nodes in $G'$ and $c'\ge 1$.

Then, for any given $\sigma\ge 1$, algorithm $\cA^*$ (Algorithm~\ref{alg:emulator}) with $\cA$ and $c$ as input
\begin{itemize}
\item    
satisfies the specification of problem $\cP$, and
\item 
(each node) terminates in $\uptau_{\cA^*}(n) = O(\uptau_\cA(n) \log n)$ 
\end{itemize}
\mglast{with probability at least $1-n^{-c'}-\uptau_\cA \cdot n^{1-c\log(4/e)}$} 
in any WBN $G$ such that $\overline{G}_\sigma = G'$.
\end{theorem}

The main technical result is to show that for same fixed random source of $\cA$, the history of $\cA$ and the produced quasi-history of $\cA^*$ executed in parallel on $\overline{G}_\sigma$ and on $G$, are the same with high probability. 

\begin{lemma}
\label{lem:emulation-histories}
Consider any WBN graph $G$ of $n$ nodes, and assume that $G_\sigma$ has no isolated node. 
Consider parallel executions of $\cA$ on $\overline{G}_\sigma$ and $\cA^*$ on $G$, with the same fixed values of random bits whenever applying function $\cA(\cdot)$ on the same (quasi-)history string.
Let $\cH_{v,\cA},\cH^*_{v,\cA^*}$ be the history vectors produced at each node $v$ in both these executions, respectively, and let $\tau$ 
\mglast{be the length of} the longest of these vectors \footnote{\mgl{The length of the suffix of undefined steps could be different between history and quasi-history.}}.
Then, the probability that for every node $v$, $\cH^*_{v,\cA^*}=\cH_{v,\cA}$, 
\mglast{is at least $1-\tau \cdot n^{1-c\log(4/e)}$,} 
where the randomness comes from selection of sequences $\alpha_v$ by $\cA^*$.
\end{lemma}

\begin{proof}
The lemma follows directly from the following invariant event (and the union bound applied to this event over all $n$ nodes), which we will prove
by induction on integer index $t\ge 0$: 
for every node $v$, either $\cH^*_{v,\cA^*}[0,\ldots,t] = \cH_{v,\cA}[0,\ldots,t]$ or both $\cH^*_{v,\cA^*}[t], \cH_{v,\cA}[t]$ are not defined, 
\mglast{with probability at least $1-t \cdot n^{1-c\log(4/e)}$.}

$\cH^*_{v,\cA^*}[0] = \cH_{v,\cA}[0]$ holds with probability $1$, independently on the graphs $G,\overline{G}_\sigma$, because $\cH_{v,\cA}[0]$ is the initial (input) state of $\cA$ and algorithm $\cA^*$ sets $\cH^*_{v,\cA^*}[0] \gets \cH_{v,\cA}[0]$. 

Suppose that the invariant holds for some $t\ge 0$; we now prove that it holds also for $t+1$.
Observe that if the ``While'' loop (line~\ref{line:loop-while}) of the algorithm continues to the next value of $index$, i.e., from $index =t$ to $index =t+1$ (line~\ref{line:next-t}), it means that $\cH_{v,\cA}[t+1]$ is well-defined -- otherwise, by correctness of algorithm $\cA$, it would decide to stop for an earlier value of $t$, so as the ``While'' loop in line~\ref{line:loop-while}.
Correspondingly, $\cH^*_{v,\cA^*}[t+1]$ can be defined only inside the ``While'' loop for $t+1$, specifically in 
line~\ref{alg-item:emulation-update}. 
Thus, either both, $\cH^*_{v,\cA^*}[t+1]$ and $ \cH_{v,\cA}[t+1]$, are well-defined (which is equivalent to satisfying condition of the ``While'' loop in line~\ref{line:loop-while} and progressing to value $t+1$ in line~\ref{line:next-t}), or both are not defined.

In the latter case, they are also not defined for all subsequent indices $t+2,t+3,\ldots$, as this is the property of $\cH_{v,\cA}$ followed from correctness of $\cA$ while $\cH^*_{v,\cA^*}$ simply stops being defined by algorithm $\cA^*$.

Consider now the former case of both $\cH^*_{v,\cA^*}[t+1], \cH_{v,\cA}[t+1]$ being defined, in which we have to prove that they are equal. There are three possible sub-cases:

\smallskip

\noindent
\underline{Sub-case $\cH_{v,\cA}[t+1]=beep$.}
It may happen only if $\cA(\cH_{v,\cA})= beep$.
By invariant for $t$, which we assumed that holds in our inductive proof, we get that this sub-case occurs if $\cA(\cH^*_{v,\cA^*}) = beep$. It follows that in the ``Else'' part of line~\ref{alg-item:emulation-update} of the algorithm $\cA^*$ for $t+1$, node $v$ sets that $\cH^*_{v,\cA^*}$ 
to $\cH_{v,\cA}$ at the considered position $t+1$, as intended.

\smallskip

\noindent
\underline{Sub-case $\cH_{v,\cA}[t+1]=1$.}
It may happen if 
$\cA(\cH_{v,\cA})= \mgl{listen}$
and there is a neighbor $u$ of $v$ in $\overline{G}_\sigma$ such that $\cA(\cH_{u,\cA}) = beep$.
By invariant for $t$, which we assumed that holds in our inductive proof, we get that this sub-case occurs if 
$\cA(\cH^*_{v,\cA^*}) = \mgl{listen}$
and there is a neighbor $u$ of $v$ in $\overline{G}_\sigma$ such that $\cA(\cH^*_{u,\cA^*}) = beep$.

If there is a neighbor $u$ of $v$ in $\overline{G}_\sigma$ such that $\cA(\cH^*_{u,\cA^*}) = beep$, then node $u$ beeps according to its sequence $\alpha_u$ when executing the ``For'' loop (line~\ref{line:for-loop}) for $t\gets t+1$. Since $u$ and $v$ are neighbors in $\overline{G}_\sigma$, the weight between them is at least $\sigma$, therefore $v$ hears beeps in all steps when $u$ beeps, in particular, during the abovementioned execution of the ``For'' loop (line~\ref{line:for-loop}) by node $u$.
Each sequence $\alpha$, in particular $\alpha_u$, has exactly half positions set to $1$, therefore $u$ beeps during half of the steps $i$ of the considered execution of ``For'' loop. Hence, node $v$ hears beeps in these steps and sets $\beta_v[i]$ to $1$ in those steps $i$ of the considered ``For'' loop execution (under ``then'' or ``Else'' of line~\ref{line:beep-and-record}, depending on $v$ beeping or not in step $i$ of the ``For'' loop). Consequently, the condition in the subsequent line~\ref{alg-item:emulation-update} is satisfied at node $v$, and it sets the value of $\cH^*_{v,\cA^*}$ to $1$ at the current position index $t+1$, which is the same as in $\cH_{v,\cA}$ in the current sub-case.

\smallskip

\noindent
\underline{Sub-case $\cH_{v,\cA}[t+1]=0$.}
It may happen if 
$\cA(\cH_{v,\cA}) = \mgl{listen}$
and there is {\it no} neighbor $u$ of $v$ in $\overline{G}_\sigma$ such that $\cA(\cH_{u,\cA}) = beep$.
By invariant for $t$, which we assumed that holds in our inductive proof, we get that this sub-case occurs if 
$\cA(\cH^*_{v,\cA^*}) = \mgl{listen}$
and there is no neighbor $u$ of $v$ in $\overline{G}_\sigma$ such that $\cA(\cH^*_{u,\cA^*}) = beep$.

If there is no neighbor $u$ of $v$ in $\overline{G}_\sigma$ such that $\cA(\cH^*_{u,\cA^*}) = beep$, then none of the neighboring nodes of $v$ beeps 
when executing the ``For'' loop (line~\ref{line:for-loop}) for $t\gets t+1$. 
Therefore, $v$ may hear a beep in a step $i$ of this execution of the ``For'' loop only if 
\[
X_i = \sum_{z\notin \{v\}\cup N_{\overline{G}_\sigma}(v)} \alpha_z[i] \cdot w(z,v) \ge 1
\ .
\]

\remove{
The expected value of $X_i$ is upper bounded by 
\[
\sum_{z\ne v} \mathbf{E}(\alpha_z[i]) \cdot w(z,v) 
=
\sum_{z\ne v}  \frac{1}{6I_\sigma}\cdot w(z,v) 
=
 \frac{1}{6I_\sigma}\cdot \sum_{z\ne v}  w(z,v) 
\le 
\frac{1}{6} 
\ ,
\]
where the last inequality follows from the definition of $I_\sigma$.
By Chernoff bound applied to the sum of $X_i$, over $1\le i\le 2c\log n$, 
%
\mgl{
we get that 
$Pr(X\geq c\log n) \leq n^{-(c+0.5)}$,
where $X=\sum_{i=1}^{|\alpha_v|} X_i$.
Then, by union bound, the probability of \emph{some} node hearing beeps more than half of the iterations of the loop in line~\ref{line:for-loop} in \emph{any} of the iterations of the loop in line~\ref{line:loop-while} is at most $\tau n^{-c+0.5}$. Thus, the claim follows.
}

\mglast{\vspace{.1in}----------------------- [[REPLACE THE LAST PARAGRAPH WITH THE FOLLOWING FOR ARXIV]]}
}

\dk{Recall that $\mathbf{E}(\alpha_z[i]) =  \frac{1}{6I_\sigma}$, by line~\ref{line:alpha} of Algorithm~\ref{alg:emulator}.}
Then, the expected value of $X_i$ is upper bounded by 
\[
\sum_{z\notin \{v\}\cup N_{\overline{G}_\sigma}(v)} \mathbf{E}(\alpha_z[i]) \cdot w(z,v) 
=
\sum_{z\notin \{v\}\cup N_{\overline{G}_\sigma}(v)}  \frac{1}{6I_\sigma}\cdot w(z,v) 
=
 \frac{1}{6I_\sigma}\cdot \sum_{z\notin \{v\}\cup N_{\overline{G}_\sigma}(v)}  w(z,v) 
\le 
\frac{1}{6} 
\ ,
\]
 %
%
where the last inequality follows from the definition of $I_\sigma$.

Define now a random variable $Y_i$ indicating whether $X_i\geq 1$ or not. That is, $Y_i=1$ if $X_i\geq 1$, or else $Y_i=0$. It follows that $\mathbf{E}(Y_i)\leq \mathbf{E}(X_i)\leq 1/6$. Let $Y=\sum_{i=1}^{|\alpha_v|} Y_i$ and $\mu = \mathbf{E}(Y)$. Then, by linearity of expectation it is 
$\mu \leq |\alpha_v|/6$.
\mglast{Notice that, for each node $v$, the variables $X_i$ (hence, $Y_i$) are independent because each node chooses each bit independently.\footnote{The dependencies among nodes and iterations are folded later in the proof in the union bound.}
Then, by the following Chernoff bound (see Thm. 7 in~\cite{DILLENCOURT2025106516}), 
$$Pr(Y > R) < 2^{-xR} \ , \ \ \ \textrm{ for } x>0, \ \mu>0, \textrm{ and } R\geq (2^{x}e-1) \mu\ ,$$
after fixing $x = \log (4/e)>0$ and $R=|\alpha_v|/2 \geq 3 \mu$,
the probability that in more than half of the 
$|\alpha_v|$ \rounds node $v$ hears a beep is
\begin{align*}
Pr(Y > |\alpha_v|/2) &< 2^{-\log (4/e) |\alpha_v|/2}
= 2^{- c \log (4/e) \log n}
= n^{- c \log (4/e)}.
\end{align*}
Then, by union bound, the probability of \emph{some} node hearing beeps more than half of the iterations of the loop in line~\ref{line:for-loop} in \emph{any} of the iterations of the loop in line~\ref{line:loop-while} is at most $n\cdot\tau\cdot n^{- c \log (4/e)} = \tau n^{1-c \log (4/e)}$. Thus, the claim follows.
}

%
%
\end{proof}

\begin{proof}[Proof of Theorem~\ref{thm:emulation}]
Let $\cA$ satisfy the assumptions of the theorem. Fix $\sigma \ge 1$.
Consider two executions: an execution of algorithm $\cA$ on a BN $G'$ of $n$ nodes, none of them being isolated, and an execution of algorithm $\cA^*$ on a WBN $G$ such that $\overline{G}_\sigma=G'$.
By Lemma~\ref{lem:emulation-histories}, for any fix source of randomness used in the function $\cA(\cdot)$, the history vector $\cH_{v,\cA}$ associated with the former execution and the quasi-history vector $\cH^*_{v,\cA^*}$ outputted in the latter are identical 
\mglast{with probability at least $1-\uptau_\cA n^{1-c\log(4/e)}$,} 
for any node $v$ in $\overline{G}_\sigma=G'$ (and thus, also in $G$, as the sets of nodes in all the three graphs are identical). 

We can now consider random sources that lead to satisfying property $\cP$ by algorithm $\cA$ on a given BN $G'$ -- \dk{by the theorem assumption, they occur with probability at least $1-n^{-c'}$.} The same sources guarantee property $\cP$ in the execution of $\cA^*$ on any corresponding WBN $G$ (i.e., such that $\overline{G}_\sigma=G'$), 
\mglast{minus the error probability $\uptau_\cA n^{1-c\log(4/e)}$} 
coming from Lemma~\ref{lem:emulation-histories}, which is caused by the randomness needed for generating $\alpha_v$ in algorithm $\cA^*$.
\dk{Hence, the emulation properties hold with probability at least $1-n^{-c'}-\uptau_\cA \cdot n^{1-c\log(4/e)}$.}

In particular, the number of iterations of the ``while'' loop in 
\dk{such executions}
is the same as the number of steps until termination in the former, with that probability, and is $\uptau_\cA(n)$. Consequently, 
\dk{such executions stop within}
$\uptau_\cA(n) \cdot 2c\log n = O(\uptau_\cA(n) \log n)$ \dk{beeping rounds}.
%
\end{proof}
\section{Additional Motivation and Other Related Work}
\label{sec:motivrelated}

From a theoretical perspective, we considered one of the simplest models of ad hoc environment with communication by beeping, WBN. We used a minimum set of model assumptions, but still we took into account weighted interactions and threshold function to separate hearing a beep from silence. Thus, our results could be transformed to more complex models and real datasets of ad hoc environments. Such transformations are, however, not straightforward and require further studies, depending on the complexity of the considered variant of the communication model. 
For instance, when applying our algorithm to the bee colony environment, one needs to take into account that bees communicate using complex signals called ``waggle dance''; it could be encoded into a binary sequence, see~\cite{dong2023social}, however a future study  
should take into account that not every binary sequence represents a valid 
waggle dance 
(i.e., recognized~by~bees).

\label{sec:bacteria}

Our WBN model can also be seen as an abstraction of chemical communication mechanisms used by bacteria. Specifically, in \emph{quorum sensing}~\cite{moreno2023quorum} bacteria release/sense some \emph{autoinducers} ($\equiv$ beeps) to the intercellular space. Upon reaching a threshold concentration ($\equiv$ \weight threshold), autoinducers trigger cascades of \emph{signal transduction}  which regulate processes such as biofilm formation, virulence, competence and sporulation~\cite{miller2001quorum, fuqua1994quorum, waters2005quorum}. Cooperation among bacteria as a multi-cellular (multiagent) system has been here for billions of years. A connection between WBN and bacteria colonies may lead to a better understanding of their collective intelligence, as well as interactions with surrounding cells of other organisms.
For instance, next step could be to extend the WBN model to different types of beeps in order to model even more complex signal transduction processes (see~\cite{galperin2004bacterial}).

\subsection{Other Related Work}
\label{sec:other-related}
\mgl{ 
\textbf{BN Research Landscape Overview.}
%
In the baseline BN model~\cite{cornejo2010deploying} nodes communicate in discrete, synchronous rounds by choosing one of two actions: beep (emit a one-bit blast of energy), or listen (detect the status of the channel).
The channel applies a logical OR operation to all beeps, that is, a listening node cannot count how many neighbors beeped; it simply senses a `1' (at least one neighbor beeped) or a `0' (complete silence). Furthermore, a node that beeps receives no feedback and cannot tell if its neighbors were also beeping at the same time.

To map the landscape of theoretical computer science research in this area, beeping models can be classified into five major families based on how they handle wake-up schedules, feedback, knowledge of parameters, noise, and network structural scale.

\begin{itemize}

\item \underline{Wake-Up and Synchronization Paradigms.}
Because many real-world sensor or biological networks do not activate simultaneously, the manner in which nodes wake up heavily impacts time complexity.

\begin{itemize}

\item Global Synchronous Wake-Up: All nodes wake up and execute the first round of the algorithm simultaneously, sharing a unified global clock~\cite{davies2023optimal,DBLP:conf/esa/GarncarekKKM25}. This is the cleanest theoretical variant and is primarily used when assessing the pure communication bottleneck of the OR channel.

\item Adversarial (Asynchronous) Wake-Up: Nodes are dormant until a central adversary decides to wake them up~\cite{feldmann2020time,dufoulon2020can}. In this harsh setting, computing basic structures like a Maximal Independent Set (MIS) in sub-polynomial time is mathematically impossible without additional assumptions~\cite{afek2011biological,hounkanli2020global}.

\item Wake-on-Beep: A dormant node is automatically shocked awake if any of its active neighbors emit a beep~\cite{afek2011biological}. This relaxation bridges the gap between asynchronous deployment and synchronous execution, allowing nodes to locally propagate a ``wake-up wave'' and drastically reducing the time required to build network structures down to poly-logarithmic time~\cite{afek2011biological,forster2014deterministic,dufoulon2018beeping,afek2013beeping}.

\end{itemize}

\item \underline{Feedback and Collision Detection Variations.}
The original, standard beeping model features no sender feedback—if a node emits a pulse, its receiver hardware is occupied, rendering it entirely deaf to what happened around it during that round~\cite{afek2011biological}. Variants relax this to gauge the power of local awareness.

\begin{itemize}
\item Standard No-Feedback Model: Beeping nodes hear nothing. Listening nodes experience standard multi-sender collisions natively as a simple `1` (beep detected)~\cite{DBLP:conf/esa/GarncarekKKM25,cornejo2010deploying,davies2023optimal}.
\item Sender-Side Collision Detection: An augmented model where a node that choose to beep is still capable of detecting whether at least one other neighbor beeped concurrently~\cite{afek2011biological,afek2013beeping}. This hardware upgrade significantly speeds up BN algorithms. For example, solving an MIS problem drops from $O(\log^3 n)$ rounds down to $O(\log^2 n)$ rounds under a wake-on-beep envelope~\cite{afek2011biological}.
\end{itemize}

\item \underline{Knowledge of Network Scale.}
Theoretical models diverge based on knowledge of network characteristics a node holds prior to deployment.

\begin{itemize}
\item Zero Knowledge: Nodes possess no information regarding the global network size $n$, the maximum node degree $\Delta$, or the specific local topology~\cite{afek2011biological,afek2013beeping}. They must algorithmically obtain those parameters if needed.
\item Bounded Parameter Models: Nodes are pre-programmed with an upper bound on $n$ and/or $\Delta$~\cite{DBLP:conf/esa/GarncarekKKM25,davies2023optimal}. Knowing an upper bound on $n$ gives the nodes a fundamental yardstick to scale down random coin-flipping probabilities during node-selection algorithms.
\end{itemize}

\item \underline{Noisy Channels vs. Clean Channels.}
As the field matured, researchers shifted focus from ideal math constructs to channels plagued by disruption, categorizing models by how faults propagate.

\begin{itemize}
\item Noiseless Beeping Model: Communication proceeds without disruption. A listener hears a beep if and only if a neighbor emitted one~\cite{cornejo2010deploying,davies2023optimal,DBLP:conf/esa/GarncarekKKM25}.
\item Independent (Local) Noisy Model: The signal received by each listening node (beep or silence) is independently flipped uniformly at random with a fixed probability~\cite{ASHKENAZI2022104925}. 
This models physical receiver defects. Remarkably, recent results prove that simulating standard message-passing algorithms over an independent noisy model requires zero asymptotic round overhead compared to a noiseless channel~\cite{davies2023optimal}.
\item Correlated (Global) Noisy Model: This model~\cite{efremenko2020noisy,davies2023optimal,gilbert2017symmetry} evaluates noise on the channel itself. When a bit flips due to noise, every single node in the network receives the exact same corrupted bit, preserving a uniform transcript among the entities but distorting the true intent of the sender. Correcting for correlated noise requires a substantial, unavoidable performance overhead.
\end{itemize}

\item \underline{Topological Scale: Single-Hop vs. Multi-Hop.}
Finally, models are split by the geometry of their interference zones.

\begin{itemize}
\item Single-Hop: The underlying network topology is a complete graph ($K_n$) where a beep from any single party is instantly heard by every listening node~\cite{efremenko2020noisy,davies2023optimal,DBLP:conf/esa/GarncarekKKM25}. In a single-hop setting, non-adaptive beeping schedules map directly to classical information-theory concepts like superimposed codes~\cite{DBLP:conf/esa/GarncarekKKM25}.

\item Multi-Hop: The network is modeled as an arbitrary undirected graph $G=(V,E)$. Beeps only reach immediate neighbors, creating intricate spatial channel conflicts where nodes must actively schedule rounds to avoid stepping on neighbors deeper in the topology~\cite{davies2023optimal,DBLP:conf/esa/GarncarekKKM25,forster2014deterministic,gilbert2017symmetry}. Recent breakthroughs show that any standard deterministic CONGEST algorithm can be simulated in any arbitrary multi-hop BN with only a deterministic multiplicative overhead of just $O(\Delta^2 \text{polylog } n)$ rounds~\cite{DBLP:conf/esa/GarncarekKKM25}.
\end{itemize}

\end{itemize}

\textbf{\probfull and Related Problems.}
}
Solving \probshort is crucial for enabling effective cooperation and scalability,
\mgl{for instance,}
in complex and large-scale 
\mgl{Multi-agent Systems (MAS)}
where agents only have partial observations of the environment. By focusing on local interactions, agents can achieve coordination and collective goals without needing a full, centralized view, which is often infeasible due to computational limits or communication constraints~\cite{oroojlooy2023review}. 
\probshort allows agents to maintain neighborhood cognitive consistency, leading to improved system-level cooperation by ensuring neighboring agents' internal representations about their local environment are aligned~\cite{mao2020neighborhood}. This decentralized approach is particularly effective in Multi-Agent Reinforcement Learning (MARL) tasks, where agents can leverage local information sharing to mitigate issues like non-stationarity and partial observability.

In Neighborhood Cognitive Consistency for instance, agents maintain coordinated policies based on joint observations within their local group, which is sufficient for effective system-level cooperation in large-scale settings like packet routing and Google Football player control~\cite{mao2020neighborhood}. This local focus provides a scalable alternative to centralized training approaches, where a single node handles all agents and suffers from performance degradation as the number of agents increases~\cite{malucellineighboring}. Furthermore, in tasks like Swarm Control and other cooperative settings, \probshort can be used to define intrinsic rewards based on the prediction error of aggregated neighboring features -- a strategy called Neighborhood-curiosity-based Exploration -- which motivates agents to explore new states and discover effective cooperative partnerships~\cite{yang2024neighborhood}. Additionally, training on alternating, large subsets of agents (neighborhoods), a technique known as Large Neighborhood Search, significantly improves training efficiency and speeds up convergence in complex cooperative MARL tasks like the StarCraft Multi-Agent Challenge~\cite{chen2024marl}.
Another example is decentralized training, where \probshort can be used to average or aggregate model parameters or stochastic gradients among neighboring nodes in the communication network. Instead of communicating with a central server, each device computes its local update and then uses the \probshort subroutine to exchange and combine this information only with its one-hop neighbors.
The repeated application of this neighborhood averaging step helps the locally trained models across all devices to converge toward a consensus (a shared ``global'' model) in a distributed manner, similar to a Distributed Parallel Stochastic Gradient~Descent.


From a technical perspective, the closest related work is on link scheduling in wireless networks. The similarities include: ad hoc network structure and potential impact of one node communication to the other. However, the problem of link scheduling is focused on realization of specific known links (typically, {\em one} link per node), unlike in \mm{\probfull} where the information needs to be propagated via {\em all} (unknown) links with sufficiently large \weight. Also, the wireless communication is focused on avoiding interference, while in \mm{WBN} we want to ``strengthen'' the information beeped via links of large \weight.

The closest wireless Link Scheduling under interference related works~\cite{HWaff,FKVaff,FKVaffJournal,KVaff,Kaff,asgeirsson2017universal,HWsurvey2019}
correspond to the Signal-to-Interference-and-Noise-Ration (SINR) interference model. 
They use a few different characteristics, called affectance, and their impact to link scheduling performance, see e.g.,~\cite{HWaff}.
The {\em generalized affectance} model was introduced and used only in the context of one-hop
communication, more specifically, to Link Scheduling in~\cite{Kaff,KVaff}. 
In~\cite{KMR_fomc14,KowalskiMZjournal}, the one-hop affectance characteristic was generalized, called the \emph{maximum average tree-layer affectance},
to be applicable to multi-hop communication tasks
such as broadcast, together with another characteristic, called the \emph{maximum path affectance}.
 In a more recent summary~\cite{HWsurvey2019} of work in Wireless Networks over the last decade theoretical studies of \LS are overviewed.




%
%


\section{Experimental Evaluation}
\label{sec:experimental}

To evaluate experimentally the performance of our algorithm we have considered 
real-world social networks (extracted from nature) on top of which nano-devices could be deployed, and to complement those we also considered ad-hoc infrastructures, which could be deployed for instance for 3D scanning and learning simultaneously.

\mamr{All input datasets used in our simulations are either publicly available in~\cite{nr}, or in our anonymous private repository in~\href{https://anonymous.4open.science/r/OBLB-A8F8/}{https://anonymous.4open.science/r/OBLB-A8F8/}, where we also include all code and output files. The latter repository will be \mmr{de-anonymized} 
upon publication of this work.} 


For natural environments, we studied performance for \mmr{four} datasets: 
a population of birds~\cite{firth2015experimental}, 
a colony of ants~\cite{mersch2013ant}, 
a population of dolphins~\cite{gazda2015importance},
\mmr{and a mouse gene regulatory network derived from analyzing gene expression profiles~\cite{bansal2007infer}. 
For the latter we extracted various subgraphs (as detailed in Table~\ref{table:datasets}) 
to process the original dataset ($|E|\approx 14.5M$).}
These inputs are representative of different network sizes (in nodes and links) and weights.

Link weights in those datasets correspond to the probability or frequency of interaction for each pair of individuals. Conceptually, these weights match our definition of \weight, because an individual $a$ that meets more frequently another individual $b$ ``weights'' more on $b$. 

\mmr{These natural datasets observe an exponential distribution of weights, biased towards small probabilities/frequencies (as in many natural interactions).} Given that the threshold for \heavy links is $\sigma\geq 1$, and that the distributions are biased towards small values, we normalized all weights to a range $[0,10]$. 

\mamr{The real-world} datasets were downloaded from the public repository in~\cite{nr} 
at 
\begin{itemize}
\item Birds: \href{https://networkrepository.com/aves-wildbird-network.php}{https://networkrepository.com/aves-wildbird-network.php}
\item Ants: \href{https://networkrepository.com/insecta-ant-colony6-day04.php}{https://networkrepository.com/insecta-ant-colony6-day04.php}
\item Dolphins: \href{https://networkrepository.com/mammalia-dolphin-florida-forage.php}{https://networkrepository.com/mammalia-dolphin-florida-forage.php}.
\mmr{\item Genes: \href{https://networkrepository.com/bio-mouse-gene.php}{https://networkrepository.com/bio-mouse-gene.php}.}
\end{itemize}

\remove{
\begin{table}[h!]
\centering
\begin{tabular}{l|c|c|c|c||c|c}
&&&&&\multicolumn{2}{c}{simulation parameters}\\
dataset&$|V|$&original $|E|$&original weights&normalized to range&$\sigma$&$w_{limit}$\\
\hline
birds~\cite{firth2015experimental,nr}&$202$&$11900$&half-weight index in $[0,1]$&$[0,10]$&$5,4.8,4.6,\dots,1$&$5,6,\dots,10$\\
ants~\cite{mersch2013ant,nr}&$164$&$10300$&interaction count in $[1,97]$&$[0,10]$&$3,2.8,2.6,\dots,1$&$4,5,\dots,10$\\
dolphins~\cite{gazda2015importance,nr}&$190$&$1100$&interaction count in $[1,7]$&$[0,10]$&$5,4.8,4.6,\dots,1$&$6,7,\dots,10$\\
grid&$216$&$1078$&exp. dist. in $[0,1]$&$[0,2^i], i=0,1,\dots,7$&$1.5,1.3,1.1$&none\\
\end{tabular}
\caption{Networks evaluated. Half-weight index $=$ probability that two individuals are observed together given that one has been seen. Interaction count $=$ number of times in spatial proximity within the time span of data collection.}
\label{table:datasets}
\end{table}
}

\begin{table}[h!]
\centering
\begin{tabular}{l|c|c|c|c}
&&&&\\
dataset&$|V|$&$|E|$&original weights&\begin{tabular}{c}normalized\\ to range\end{tabular}\\
\hline
\hline
birds:~\cite{firth2015experimental}&$202$&$11900$&half-weight index in $[0,1]$&$[0,10]$\\
\hline
ants:~\cite{mersch2013ant}&$164$&$10300$&interaction count in $[1,97]$&$[0,10]$\\
\hline
dolphins:~\cite{gazda2015importance}&$190$&$1100$&interaction count in $[1,7]$&$[0,10]$\\
\hline
\mmr{genes:~\cite{bansal2007infer}}&
$1000$&
$90881$&
probabilistic interaction&
$[0,10]$\\
&$2000$&
$258909$&
&\\
&$3000$&
$432239$&
&\\
&$4000$&
$563033$&
&\\
&$5000$&
$675129$&
&\\
\hline
\mmr{synthetic scale-free network}&$1000$&$3994$&exp. dist. in $[0,1]$&$[0,10]$\\
&$2000$&
$7994$&
&\\
&$3000$&
$11994$&
&\\
&$4000$&
$15994$&
&\\
&$5000$&
$19994$&
&\\
\hline
synthetic 3D grid&$216$&$1078$&exp. dist. in $[0,1]$&\begin{tabular}{c}$[0,2^i],$\\ $i=[0,7]$\end{tabular}\\
\end{tabular}
\caption{Networks evaluated. Half-weight index $=$ probability that two individuals are observed together given that one has been seen. Interaction count $=$ number of times in spatial proximity within the time span of data collection.}
\label{table:datasets}
\end{table}

\begin{table}[h!]
\centering
\begin{tabular}{l|c|c}
dataset&$\sigma$&$w_{limit}$\\
\hline
\hline
birds:~\cite{firth2015experimental}&$5,4.8,4.6,\dots,1$&$5,6,\dots,10$\\
ants:~\cite{mersch2013ant}&$3,2.8,2.6,\dots,1$&$4,5,\dots,10$\\
dolphins:~\cite{gazda2015importance}&$5,4.8,4.6,\dots,1$&$6,7,\dots,10$\\
\mmr{genes:~\cite{bansal2007infer}}&$5,4,\dots,1$&$10$\\
\mmr{synthetic scale-free network}&$1.5,1.3,1.1$&$10$\\
synthetic 3D grid&$1.5,1.3,1.1$&none\\
\end{tabular}
\caption{Simulation parameters.}
\label{table:parameters}
\end{table}

To complement our real-world scenarios, we also evaluated two ad-hoc \mmr{networks: a scale-free network and a} 3D-topology. \mmr{To generate the scale-free network we used the preferential attachment model of~\cite{barabasi1999emergence}, obtaining various input graphs with $|V|=1000, 2000, 3000, 4000, 5000$ and $|E|=3994, 7994, 11994, 15994, 19994$ respectively.
For the second we created} 
a 3D-grid of $6\times 6\times 6$ vertices placing a node in each intersection, which yields $|V|=216$ and $|E|=1078$.
Weights were defined following an exponential distribution similar to the nature datasets. Specifically, for \mmr{in the 3D grid each weight was drawn} uniformly at random a number in $(i,i+0.1]$ with probability $1/2^i$, for each $i=0,0.1,\dots,0.9$; \mmr{and for the scale-free network we used a random exponential distribution with rate parameter 5.} 

The main characteristics of the networks evaluated and the parameter values chosen are listed in \mam{Tables~\ref{table:datasets} and~\ref{table:parameters} respectively.} 

\subsection{How to \mgl{vary} $I_\sigma$}

\remove{
For the datasets evaluated, \mwd $I_\sigma$ is unknown. 
Computing $I_\sigma$ exactly could be time consuming.
Our algorithm can be extended to handle the issue implementing an exponential search of $I_\sigma$ (doubling estimates), as explained in Sec.~\ref{sec:extensions}. 
Nevertheless, to focus on the algorithmic dependency on $I_\sigma$ and $E_{\sigma}$, we developed experimentally an estimation method.
Namely,
for each input dataset, we estimated $I_\sigma$  randomly sampling the set of \heavy links in size and content. To improve accuracy, we repeated this estimate $200$ times, keeping the largest $I_\sigma$ obtained. 
An appropriate number of times to repeat the calculation, so that the chosen value ($200$) yields an accurate estimation, was experimentally determined trying increasing numbers until convergence is reached. As seen in Figure~\ref{fig:repetitions} \mmr{for some of our inputs}, the estimate converges rapidly even with a much smaller number of repetitions~\footnote{Notice that our goal in this work is not to provide a provable, and/or decentralized, estimator for the \mwd - which is an interesting open question on its own.}.

\begin{figure}[htbp]
\centering
\includegraphics[width=0.49\columnwidth]{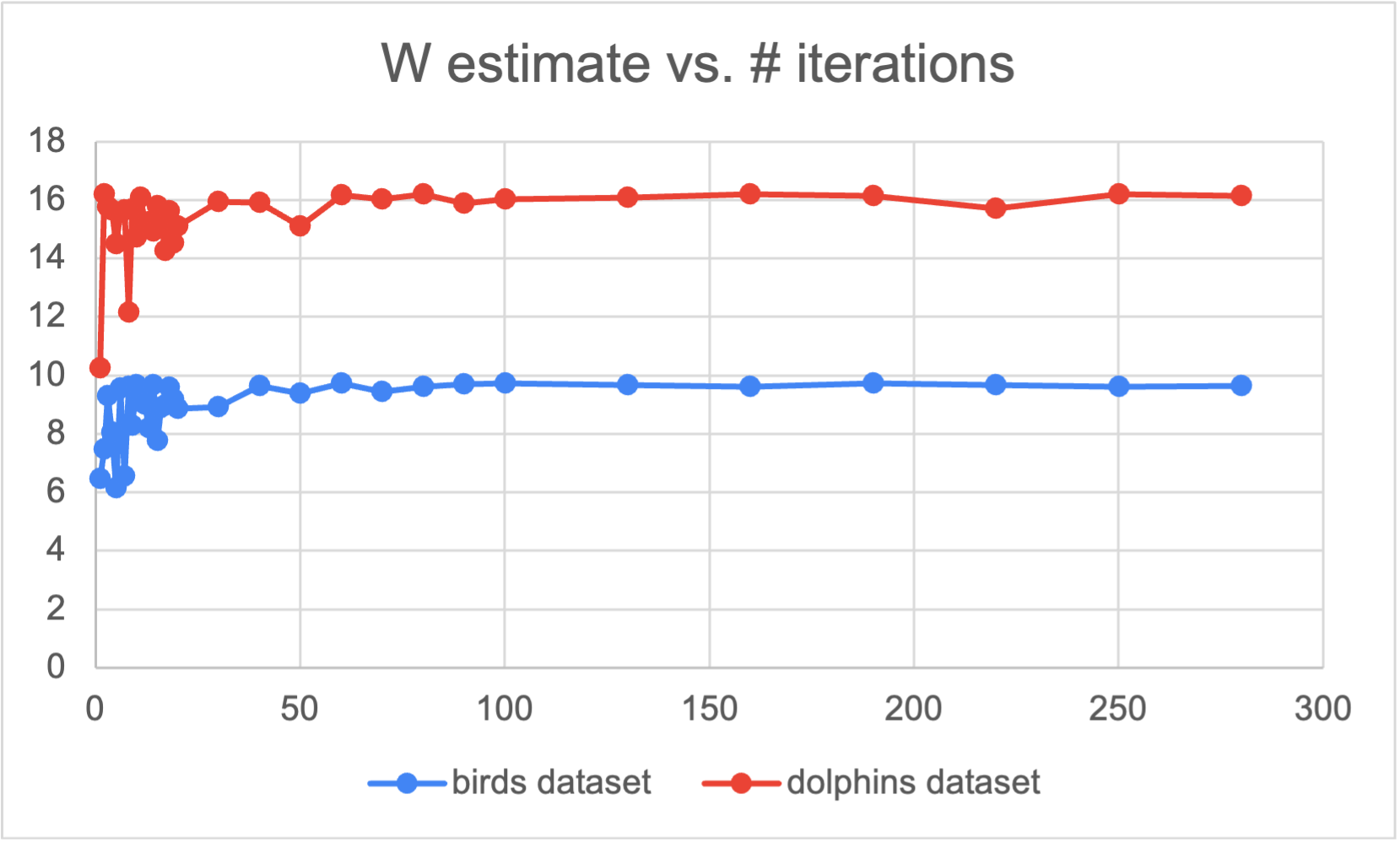}
\hfill
\includegraphics[width=0.49\columnwidth]{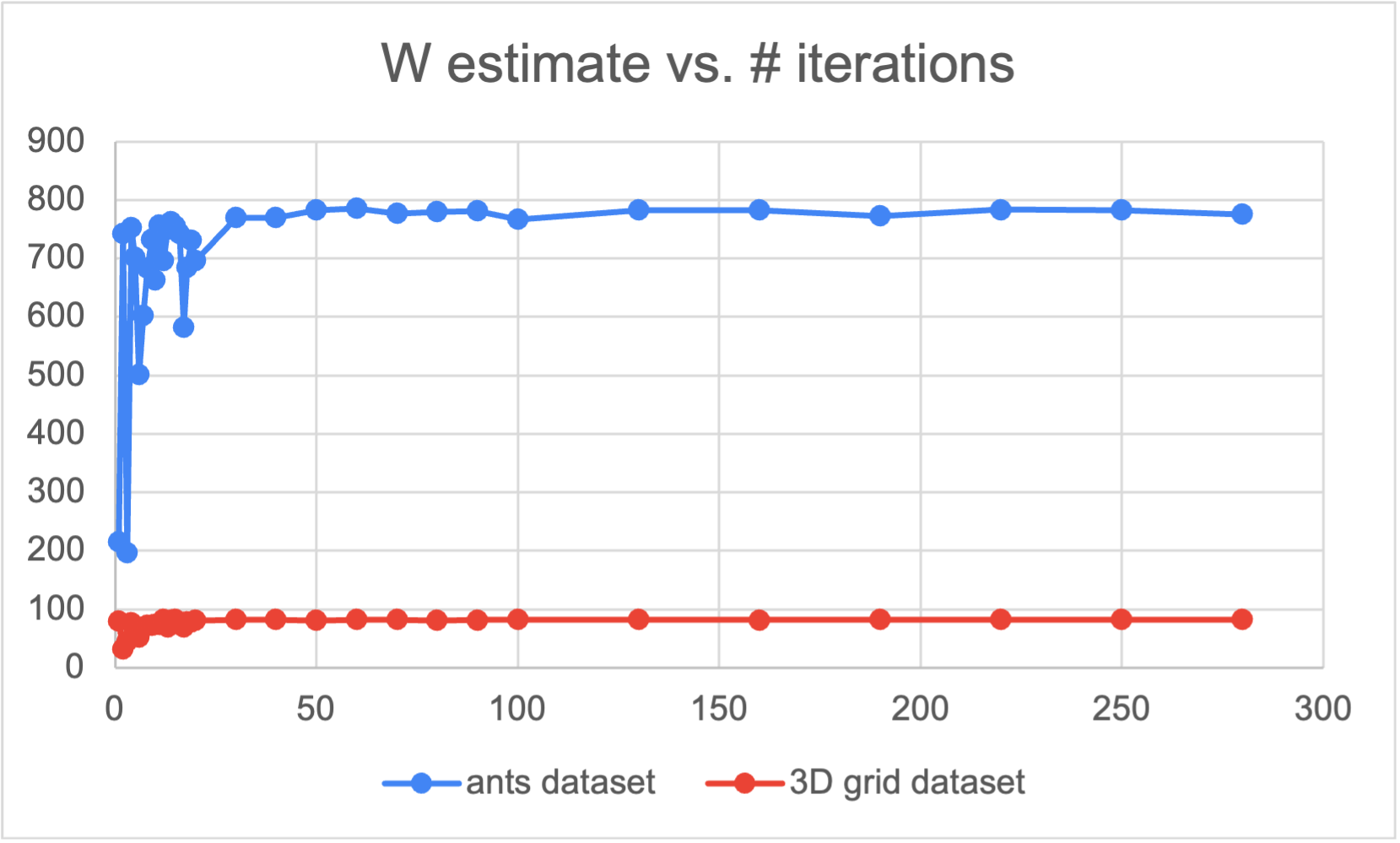}
\caption{\mwd estimate versus number of repetitions of the estimation for all datasets studied.}
\label{fig:repetitions}
\end{figure}
}

A significant challenge is how to obtain various inputs with different $I_\sigma$ values (based on the same dataset) so that we can evaluate the dependency of time on $I_\sigma$.  
\mmr{For the genes 
dataset, we addressed this challenge using decreasing values of $\sigma$ from an initial larger value ($\sigma_{limit}\leq 10$) down to $1$. A changing $\sigma$ results in a changing $|E_{\sigma}|$ (the number of \heavy links), which possibly yields a changing $I_\sigma$ (see the definition of $I_\sigma$ in Section~2.4).
For the birds, ants, and dolphins datasets, we additionally used} subsets of links of the input graph, adding only the links with weight up to some limit $w_{limit}$, from some $w_{limit} \geq \sigma_{limit}$ up to the maximum weight $10$. The latter handles the case where the distribution of weights is such that the end points of links with large weight have also large weighted degree (e.g. ants dataset). Then, starting with a large $\sigma_{limit}$ (hence, small $|E_{\sigma}|$) would not change significantly $I_\sigma$, even for larger $|E_{\sigma}|$.
For the synthetic \mmr{datasets (i.e., 3D-grid and scale-free networks)} on the other hand, we induced the increasing values of $I_\sigma$ adding a multiplicative factor of \mmr{$1,2,4,\dots,f$ to the weights, where $f=128$ for the 3D grid and $f=8$ for the scale-free networks,} for $\sigma=1.5,1.3,$ and $1.1$.
\mmr{All those ranges were determined experimentally to attain enough variety of values of $I_\sigma$.}

For the values of $I_\sigma$ obtained, we evaluated the performance of our \emph{One-beep Local Broadcast (OBLB)} algorithm as defined in Section~3.
Evaluating OBLB is enough for the purpose of evaluating the dependency of our \probfull algorithm on $I_\sigma$, because the overhead of \probfull with respect to OBLB is only logarithmic on $n$ (for messages of logarithmic size). 

\subsection{Baseline algorithms}

To the best of our knowledge, there is no other One-beep Local Broadcast (or \probfull) algorithm that runs in the beeping model. Thus, for experimental comparison with other solutions, we considered an algorithm in which nodes start execution with random delays in $[n]$, and then beep with a period equal to their own ID, which we call the \emph{Periodic} algorithm. 
We tested both algorithms on the same inputs and parameter combinations. The results are shown in Figures~\ref{fig:birdsantsplots},~\ref{fig:dolphinsgenesplots}, and~\ref{fig:synthetic}. 

One may ask why not using a simple round-robin schedule which would solve the problem in time $n$. The shortcoming of such approach of course is that for a massive network the running time would be very slow, which is intuitively unnecessary for a local computation problem such as \probfull. Moreover, a round-robin algorithm would require the nodes to know the global time and $n$. Nevertheless, the question that follows is how OBLB's running time scales with $n$ in practice. To answer that question we evaluated OBLB experimentally on a \mamr{synthetic} 3D-grid as defined above but with growing sides $6,7,\dots, 15$, thus $n=6^3,7^3,\dots,15^3$. Weights were defined with an exponential distribution on $[0,10]$ biased to small values as before. 
To maintain $I_\sigma$ and $|E_{\sigma}|$ ``stable'' under a growing $n$ we set $\sigma=1$ initially, increasing it by $0.5$ each time that we increased $n$. The result was that during the whole simulation $I_\sigma$ and $|E_{\sigma}|$ oscillated but stayed within the ranges $[10,22]$ and $[402,692]$ respectively, with $n$ growing monotonically from $216$ to $3375$. The results are shown in Figure~\ref{fig:gridnvariable}. 

To measure running time we counted \rounds of communication until 
a beep was 
\miguel{dominant}
in all \heavy links.
All executions were repeated $20$ times for each parameter combination.
The seed for all random samples was $31277847$.

\subsection{Discussion of results}
\label{sec:expdisc}
The results of our experiments show that performance of our OBLB in practice is similar or better than the theoretical bound, whereas the Periodic algorithm is worse than the same bound. We also show that OBLB running time is independent of $n$. We have illustrated our results in box-and-whisker charts to show the statistical behavior~\footnote{Although it is not random, the theoretical bound of $4I_\sigma\log|E_{\sigma}|$ sometimes varies for fixed $I_\sigma$ when $|E_{\sigma}|$ changes.} (recall that each execution was repeated $20$ times). \mmr{Notice that, in comparison with network size, our datasets comprise a variety of $I_\sigma$ values, from comparable with $|V|$ (ants) to very small (all others for original weights).}

The charts for the birds dataset in 
Figure~\ref{fig:birdsantsplots}
show on the left that, as $I_\sigma$ grows, the performance of OBLB gets much better than $4I_\sigma\log|E_{\sigma}|$ (with the upper quartile of the former smaller than the lower quartile of the latter), whereas on the right we can see that the Periodic algorithm performs similar or worse than the same function. 
Similar observations apply to the dolphins dataset in 
Figure~\ref{fig:birdsantsplots},
where in fact the difference in favor of OBLB is larger. 
\mmr{For the ants and mouse genes datasets, we can see in 
Figure~\ref{fig:dolphinsgenesplots}
much better performance of OBLB (left) than the theoretical $4I_\sigma\log|E_{\sigma}|$ time, whereas the Periodic algorithm (right) took more than $30000$ \rounds for most parameter combinations.}

\begin{figure}[h]
\centering
\includegraphics[width=0.49\columnwidth]{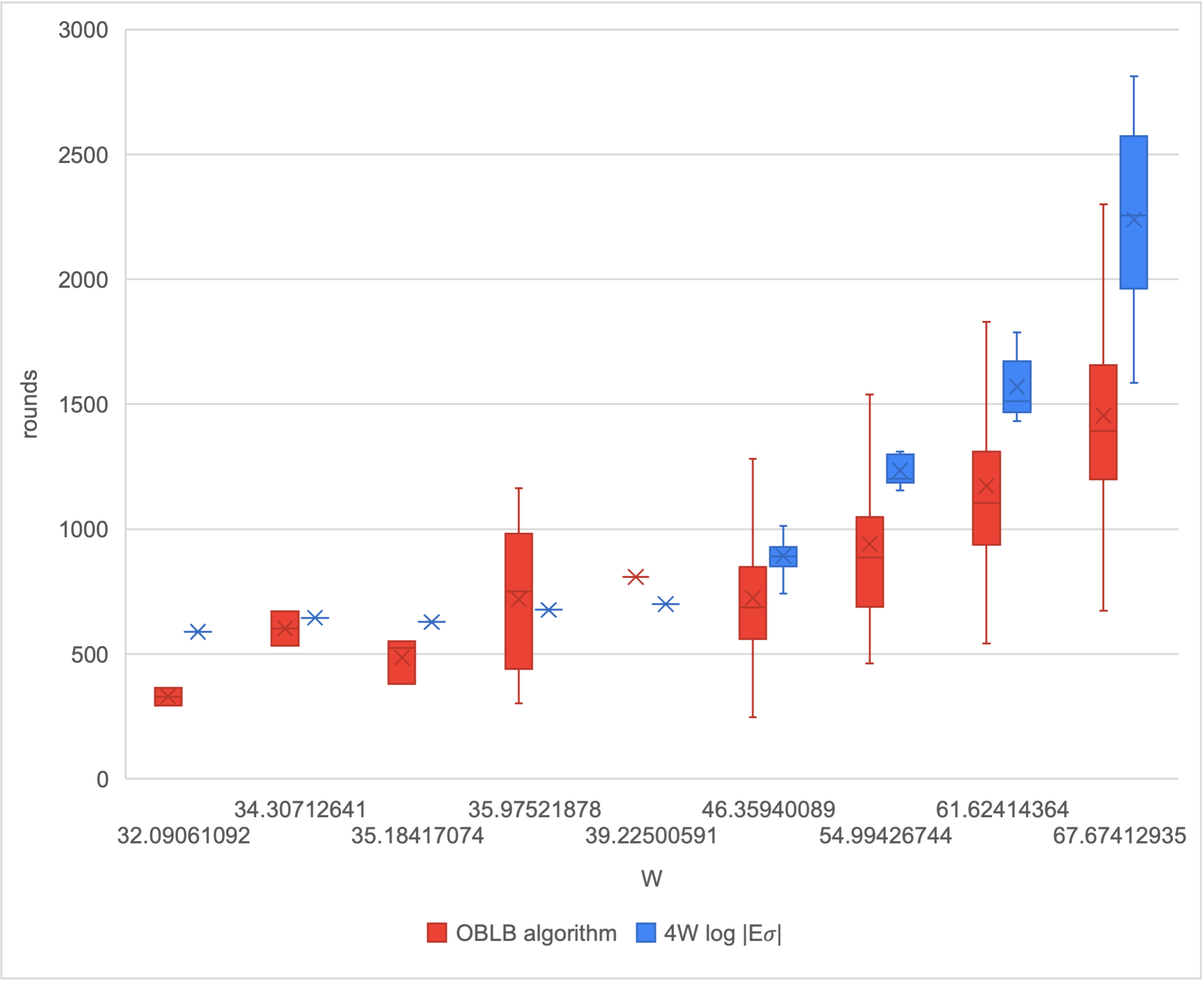}
\hfill
\includegraphics[width=0.49\columnwidth]{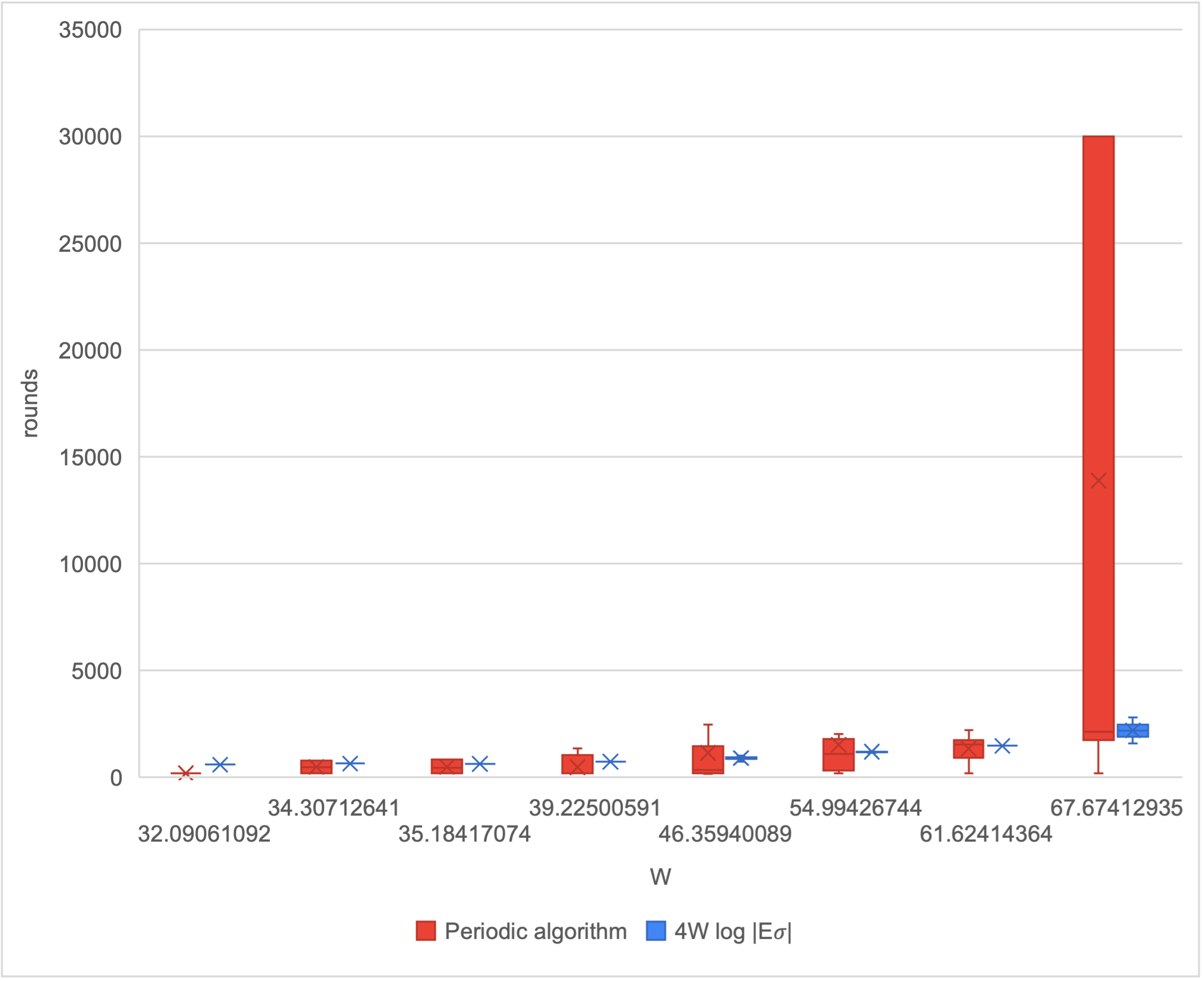}
\label{subfig:birds}
\newline
\newline
\includegraphics[width=0.49\columnwidth]{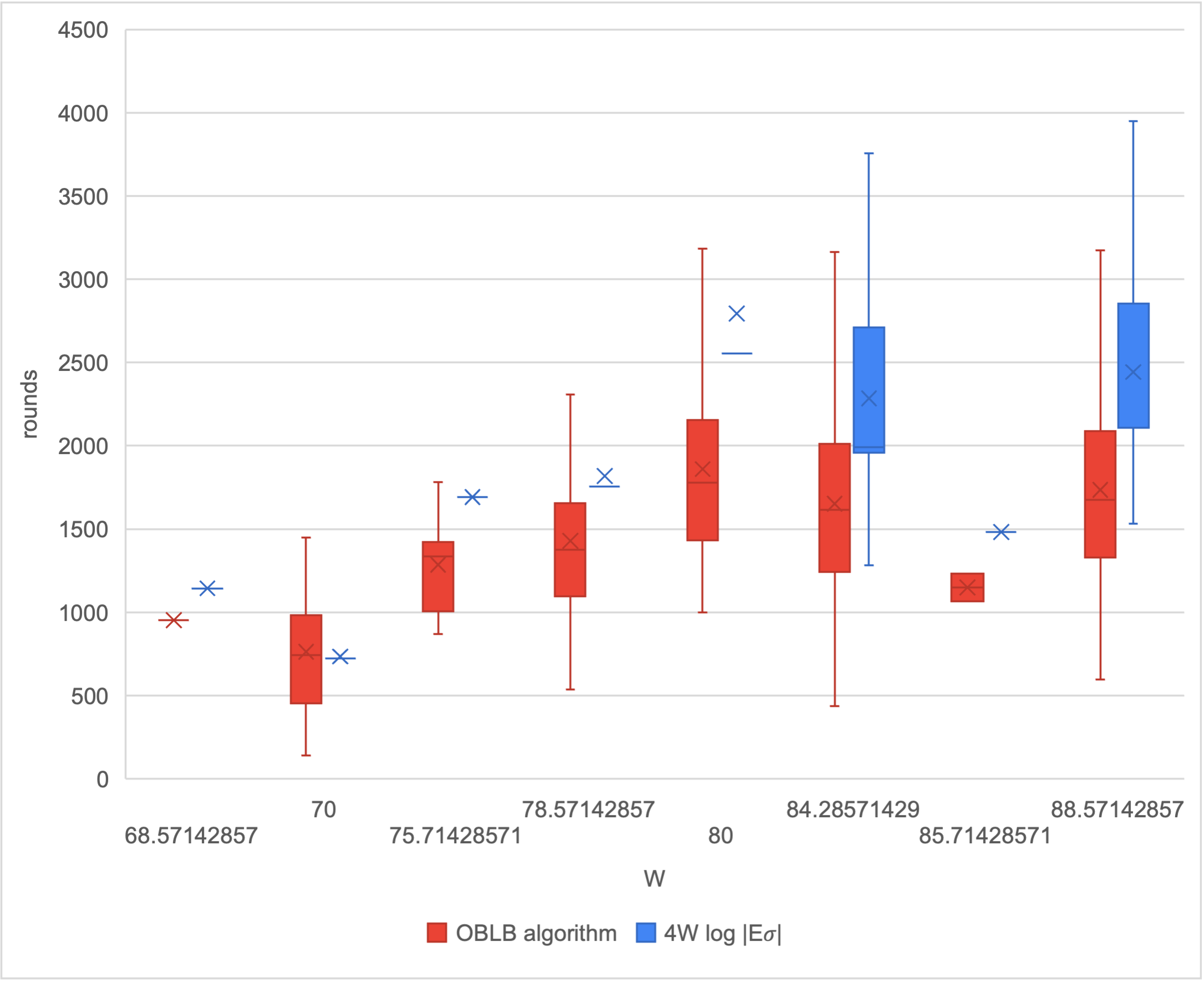}
\hfill
\includegraphics[width=0.49\columnwidth]{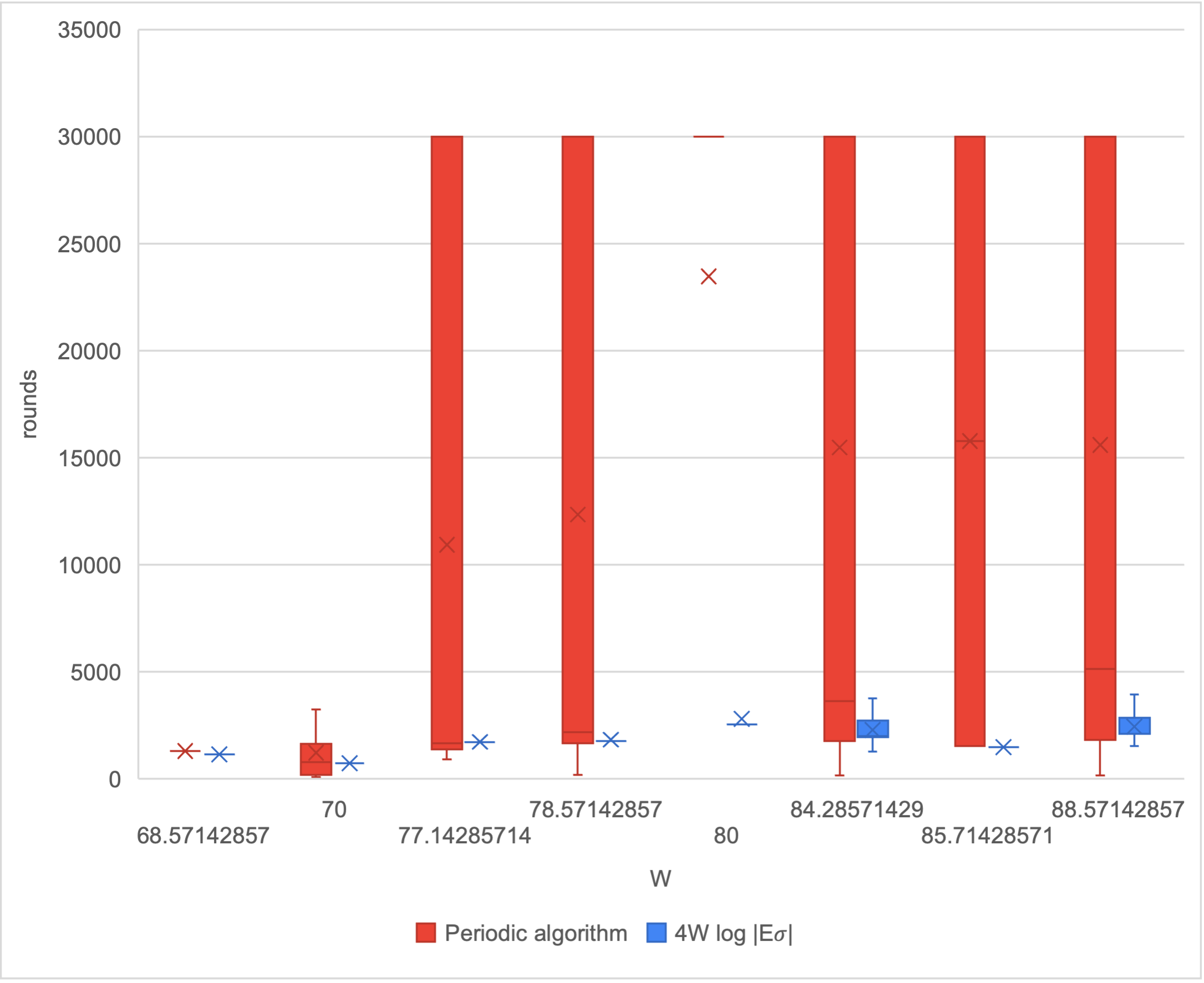}
\label{subfig:dolphins}
\caption{
Experimental results on nature-extracted datasets
in comparison with OBLB theoretical running time of $4I_\sigma\log|E_{\sigma}|$.
Left: OBLB Algorithm, right: Periodic algorithm.
\miguel{Above: Birds population dataset, below: Dolphins population dataset.}
Statistical behavior over $20$ executions for each parameter combination.
All executions capped at $30000$ \rounds. That is, datapoints shown at $30000$ \rounds in fact indicate a running time of $\geq 30000$ \rounds. 
}
\label{fig:birdsantsplots}
\end{figure}

\begin{figure}[ht]
\centering
\includegraphics[width=0.49\columnwidth]{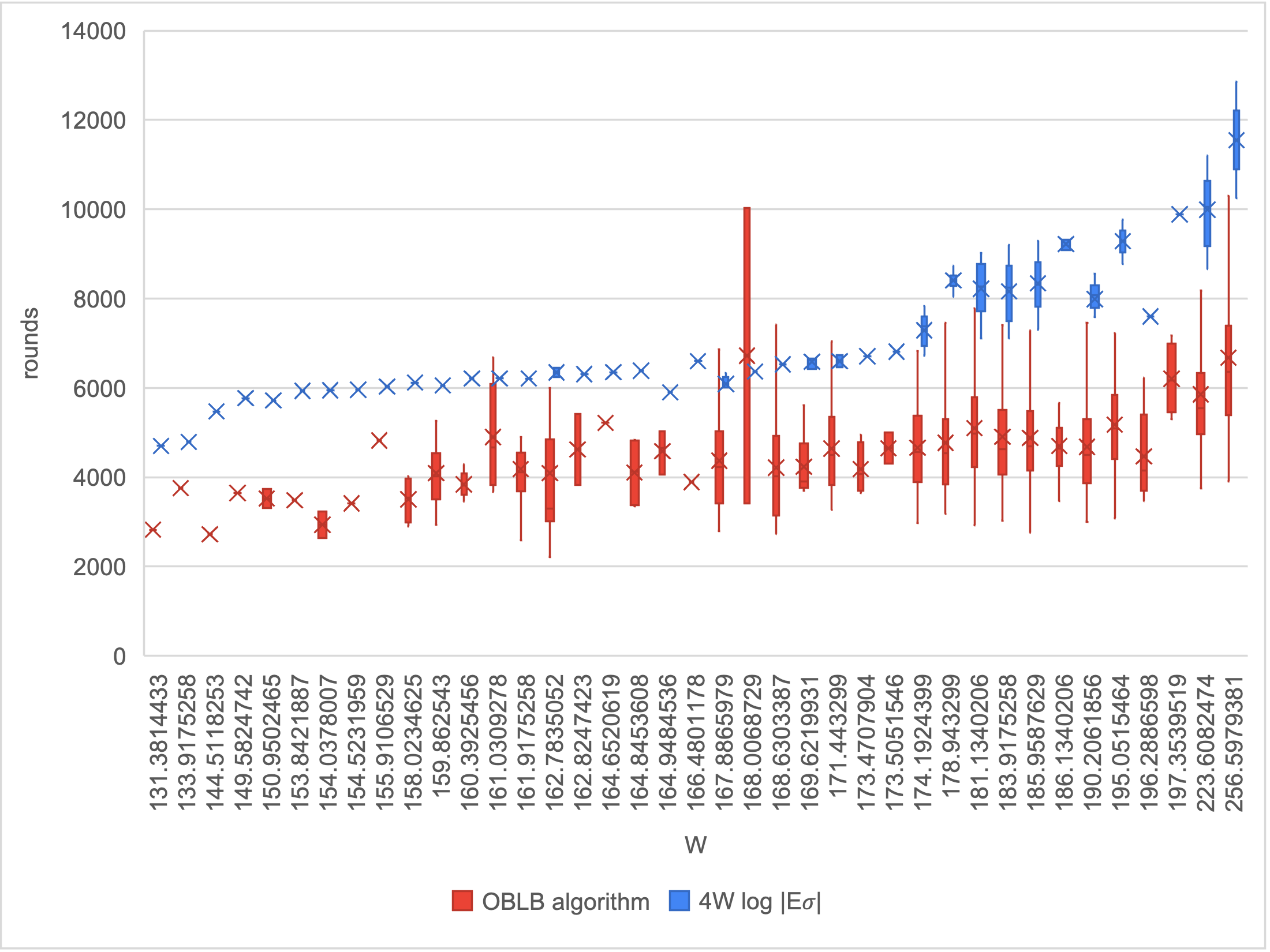}
\hfill
\includegraphics[width=0.49\columnwidth]{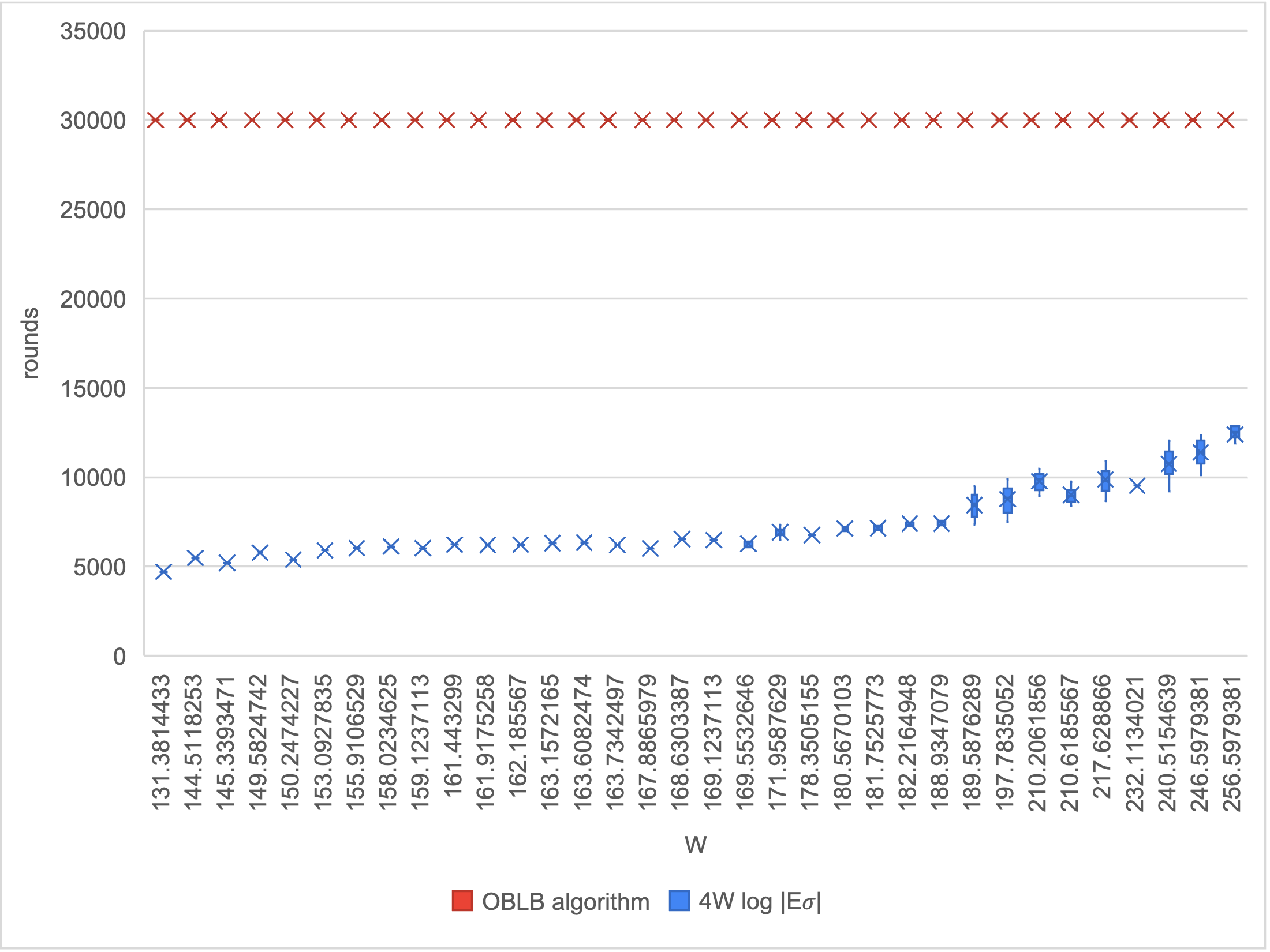}
\label{subfig:ants}
\newline
\newline
\includegraphics[width=0.49\columnwidth]{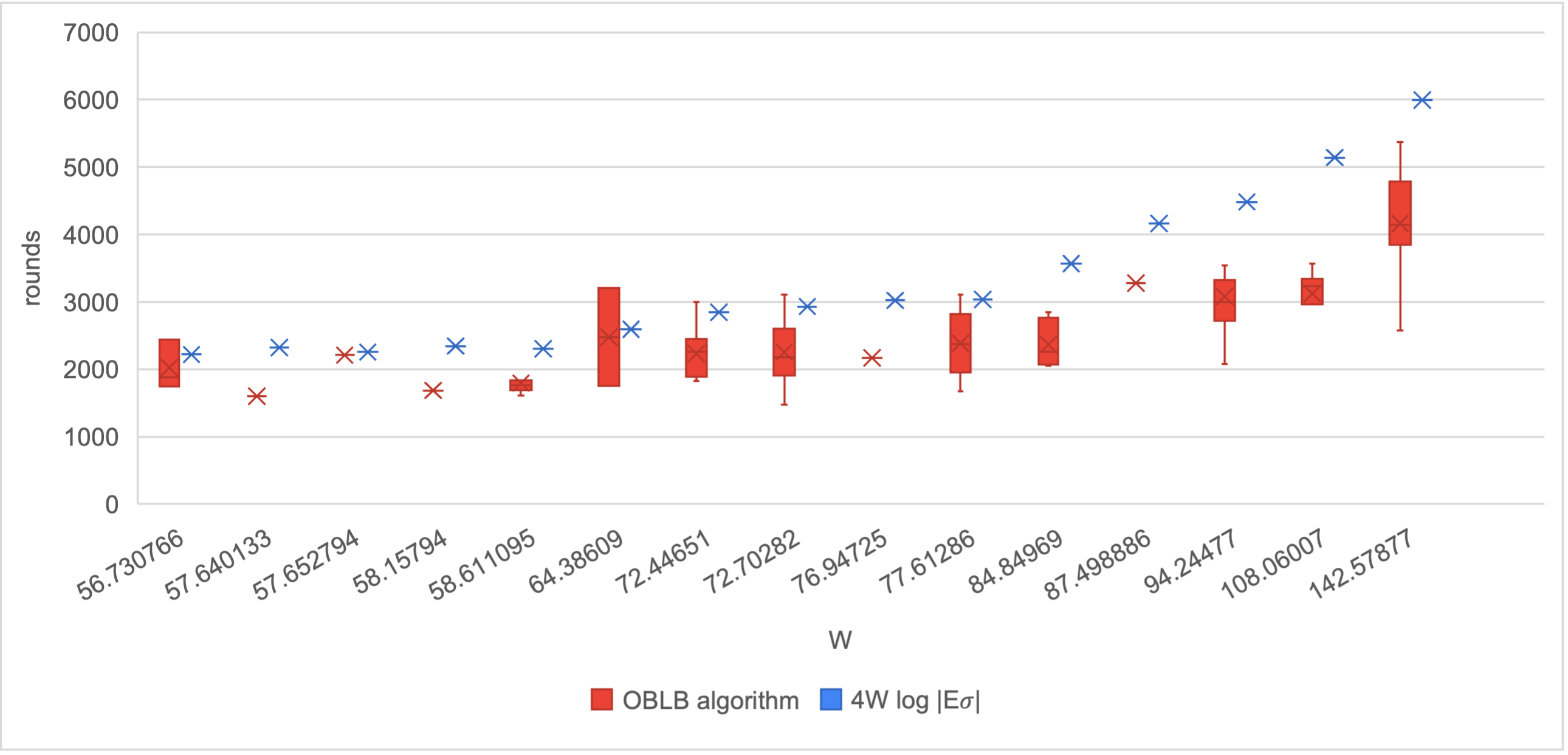}
\hfill
\includegraphics[width=0.49\columnwidth]{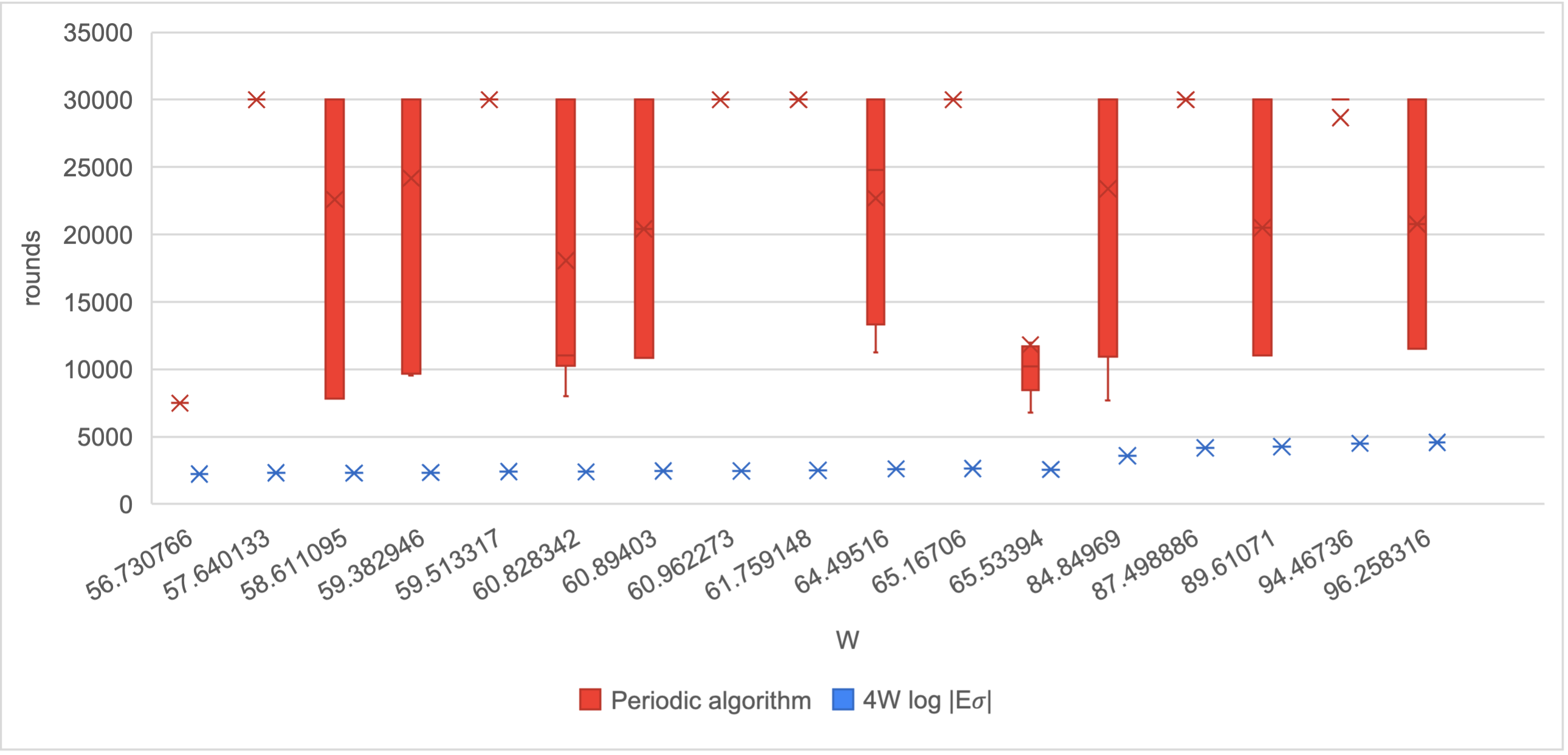}
\label{subfig:genes}
\caption{
Experimental results on nature-extracted datasets
in comparison with OBLB theoretical running time of $4I_\sigma\log|E_{\sigma}|$.
Left: OBLB Algorithm, right: Periodic algorithm.
Above: Ant colony dataset, below: Mouse genes dataset ($|V|=1000, |E|=90881$ shown, larger graphs produced similar results).
Statistical behavior over $20$ executions for each parameter combination.
All executions capped at $30000$ \rounds. That is, datapoints shown at $30000$ \rounds in fact indicate a running time of $\geq 30000$ \rounds. 
}
\label{fig:dolphinsgenesplots}
\end{figure}

\mmr{The performance on synthesized topologies, scale-free networks and 3D grids, can be seen in Figure~\ref{fig:synthetic},} where the same observations apply: the running time of OBLB is almost always below $4I_\sigma\log|E_{\sigma}|$ (with a larger difference as $I_\sigma$ grows), while the running time of the Periodic algorithm is above (many times above $30000$ \rounds, when we stopped the simulation). 

For the evaluation of the dependency of OBLB on $n$, we can see in Figure~\ref{fig:gridnvariable} that indeed it is independent (as the theoretical analysis showed). That is, for $n$ growing monotonically from $216$ to $3375$ the average running time of OBLB oscillated within a range of $[330,580]$. That oscillation is due to the oscillation of $I_\sigma$ and $|E_{\sigma}|$, which can be seen in the plot of the theoretical $4I_\sigma\log |E_{\sigma}|$ included for comparison in the same chart.

\begin{figure}[htbp]
\centering
\includegraphics[width=0.49\columnwidth]{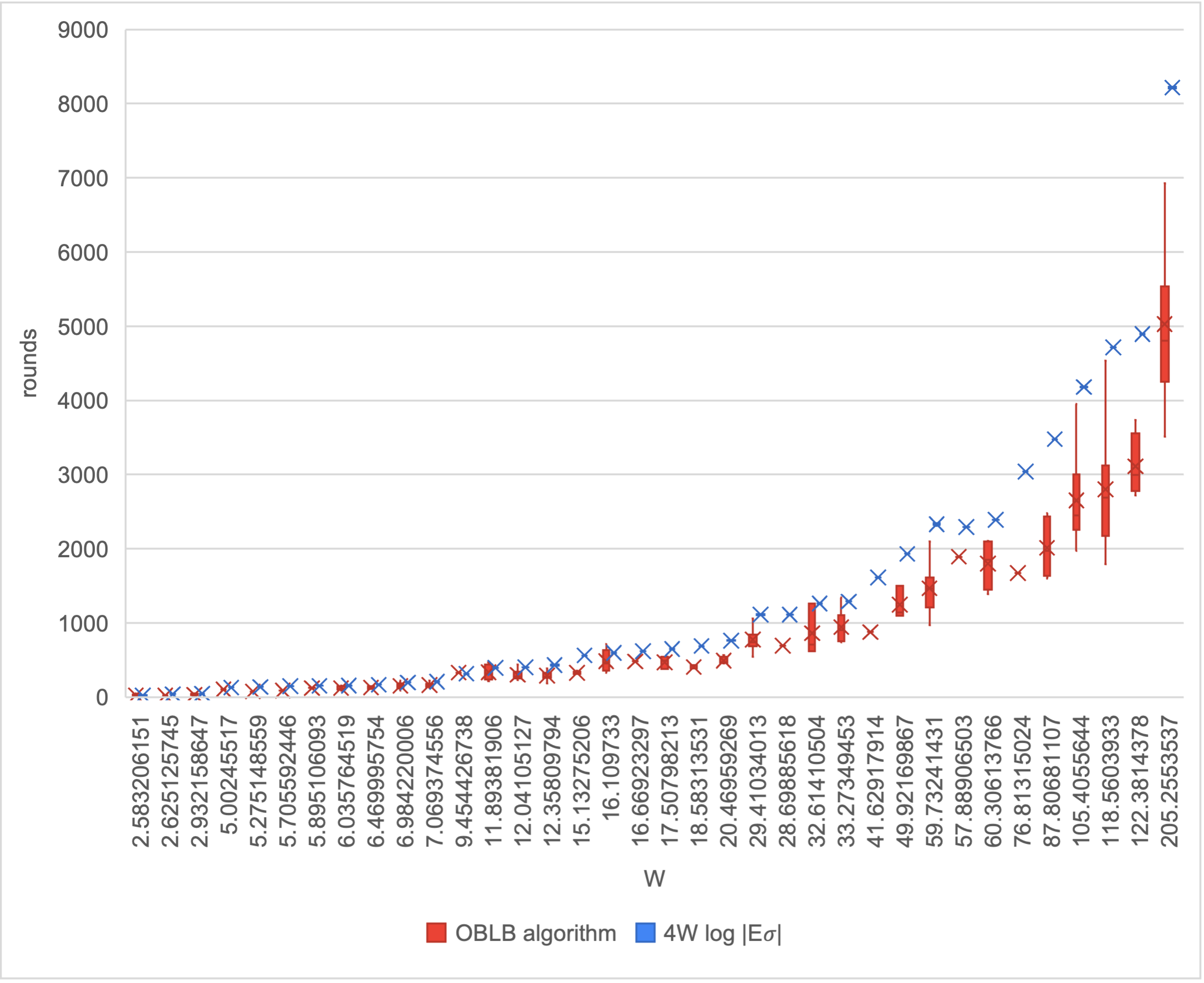}
\hfill
\includegraphics[width=0.49\columnwidth]{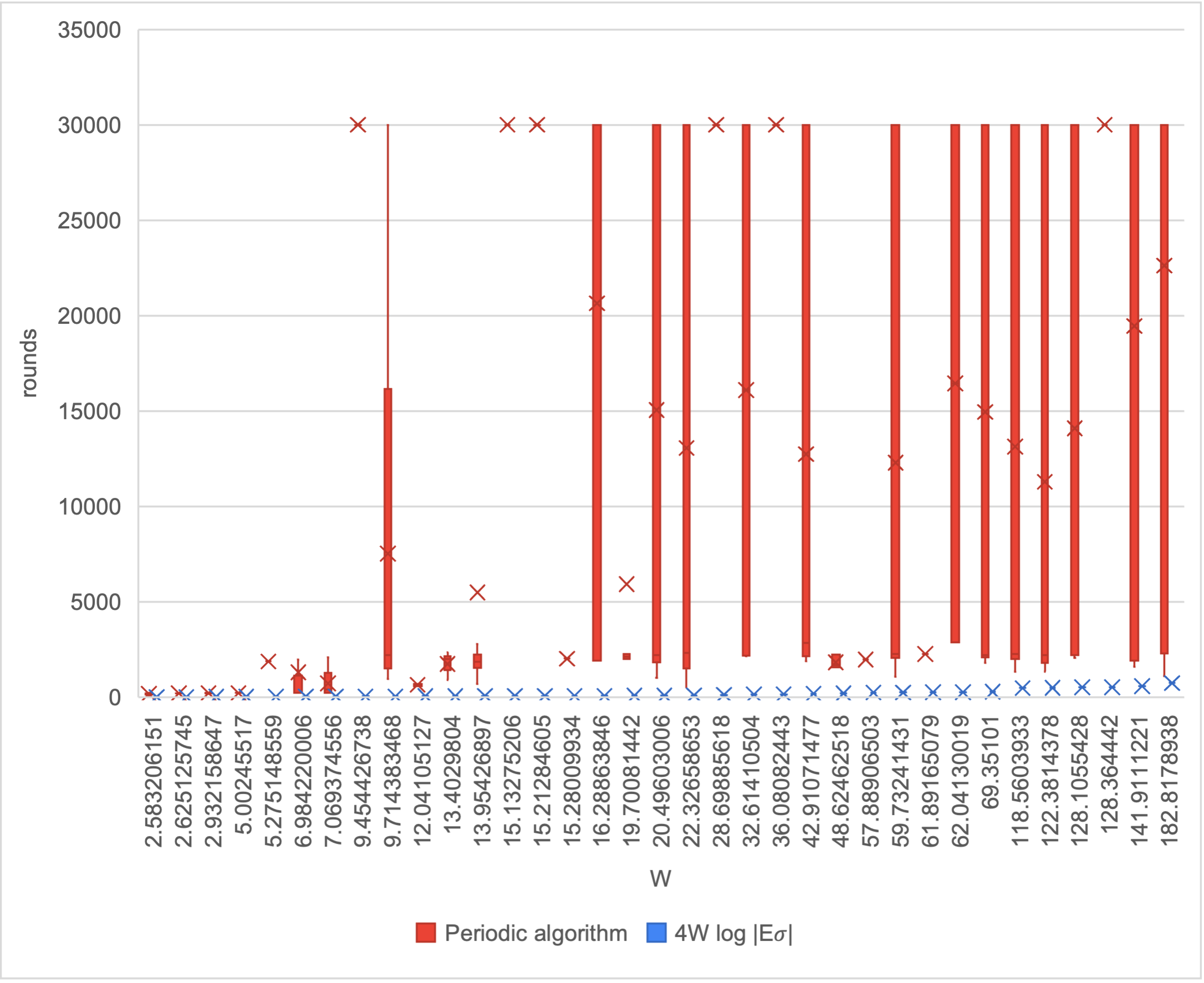}
\label{subfig:grid}
\includegraphics[width=0.49\columnwidth]{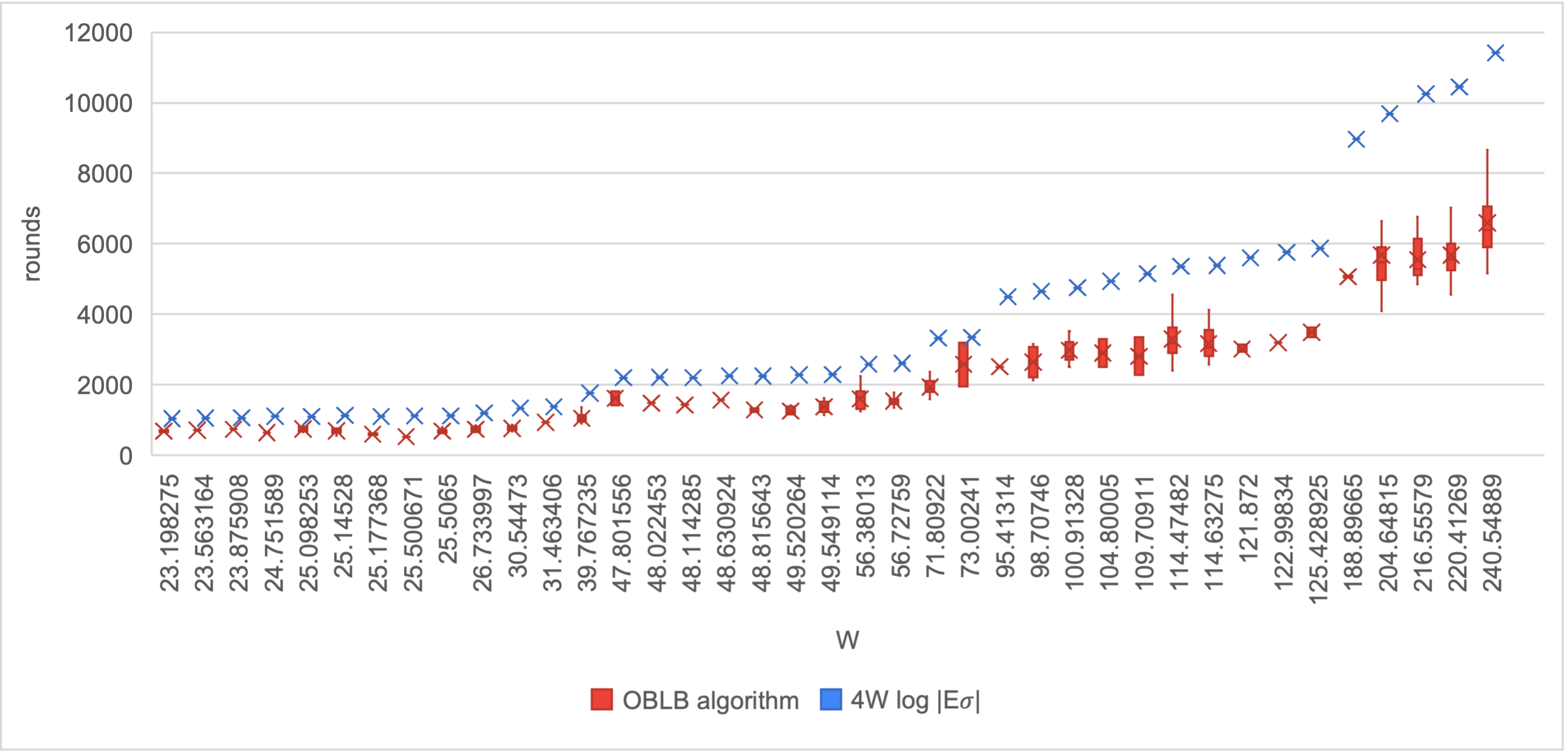}
\hfill
\includegraphics[width=0.49\columnwidth]{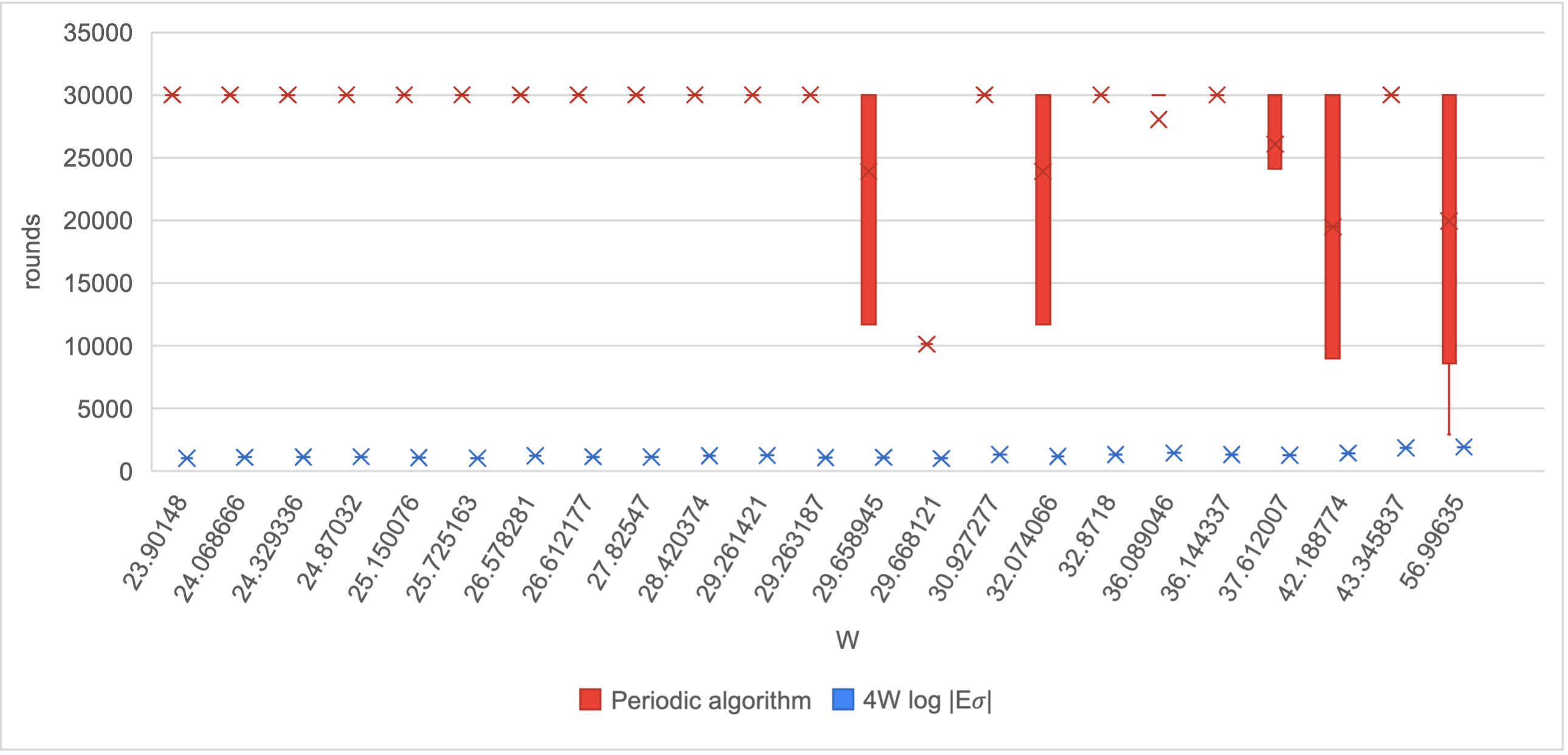}
\label{subfig:scalefree}
\caption{
Experimental results on synthetic topologies in comparison with OBLB theoretical running time of $4I_\sigma\log|E_{\sigma}|$.
Left: OBLB Algorithm, right: Periodic algorithm.
Above: 3D grid dataset, below: Scale-free network dataset ($|V|=1000$, $|E|=3994$ shown, larger graphs produced similar results).
Statistical behavior over $20$ executions for each parameter combination.
All executions capped at $30000$ \rounds. That is, datapoints shown at $30000$ \rounds in fact indicate a running time of $\geq 30000$ \rounds.}
\label{fig:synthetic}
\end{figure}

\begin{figure}[htbp]
\centering
\includegraphics[width=0.6\columnwidth]{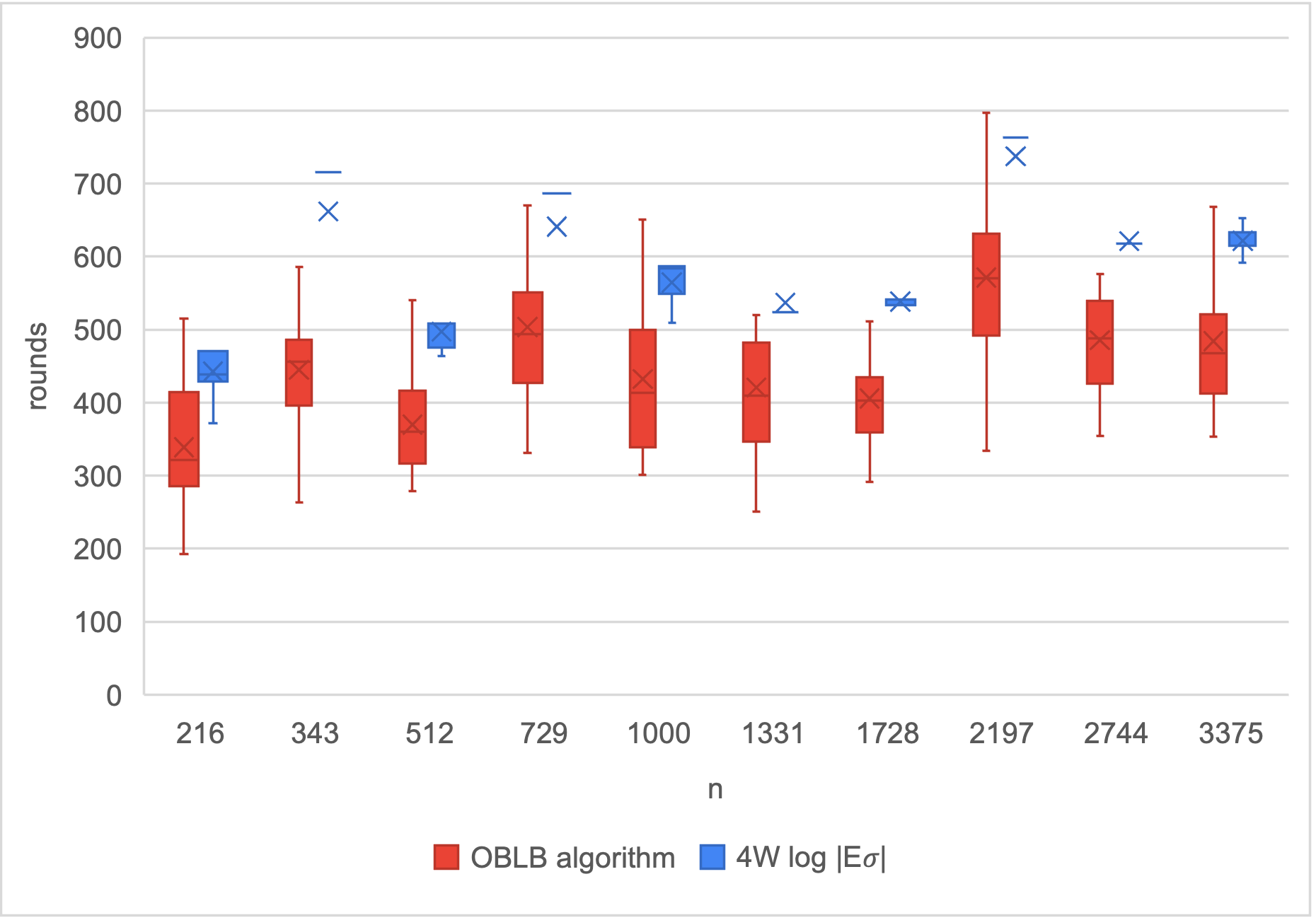}
\caption{
Experimental results on a \mamr{synthetic} 3D-grid topology showing that OBLB running time is independent of $n$. The OBLB theoretical running time of $4I_\sigma\log|E_{\sigma}|$ is included to observe that the oscillation of OBLB running time is due to the oscillation of $I_\sigma$ and $|E_{\sigma}|$.
Statistical behavior over $20$ executions for each parameter combination.
}
\label{fig:gridnvariable}
\end{figure}


\subsection{Reproducibility of Experiments}
The \mmr{nature-extracted} datasets are publicly available in the repositories cited above.
All synthesized datasets are included in an anonymous private repository at \href{https://anonymous.4open.science/r/OBLB-A8F8/}{https://anonymous.4open.science/r/OBLB-A8F8/} and will be made publicly available upon publication of the paper. 

The simulators of OBLB and Periodic algorithms were coded in Java language, compiled in the Java SE Runtime Environment (build 1.8.0\_121-b13), and executed in a Java HotSpot 64-Bit Server VM (build 25.121-b13, mixed~mode).
All code is included in an anonymous private repository at \href{https://anonymous.4open.science/r/OBLB-A8F8/}{https://anonymous.4open.science/r/OBLB-A8F8/} and will be made publicly available upon publication of the paper. 
\section{Discussion and Open Directions}
\label{sec:discussion}

We showed, using formal analysis, 
that the complexity of \probshort in WBN networks
is nearly-linearly proportional to the newly introduced characteristic -- \mwd.
Such characteristic is global, and its range may span from a constant to the number $n$ of all nodes, depending on the weighted network (even if neighborhoods are small).
Our lower bound shows that executing \probshort may be costly even in networks with small degree (but with high \mwd). On the other hand, for networks with small \mwd (which was also the case in most of the considered nature-driven and synthetic datasets,
we give a provably-efficient nearly-optimal algorithmic solution.
Our theoretical bounds achieve a gap of at most $(\cM+\log n)\log n$ between the upper bound (for messages of length $\cM$) and lower bounds (for single-bit messages). That means that our upper bound is nearly-optimal for systems with short messages $\cM\in O(\log n)$, yielding a gap of $O(\log^2 n)$. That is, our algorithm is well scalable.  
%
Our experimental analysis
confirming our theoretical results 
on diverse nature-driven and synthetic datasets. 
Our algorithm~also~outruns~other~competitors. 
%
\mgl{In this work, we also presented an emulator of any polynomial algorithm designed for a BN that runs on a WBN. Our emulator succeeds whp introducing only a $O(\log n)$ time overhead over the original BN algorithm. Thus, the wealth of BN algorithms in the literature can be applied to WBNs with only logarithmic overhead.} 

Interesting open problems include: the study of
more complex nature-inspired models in the WBN framework (including binary encoding of more complex model features), 
more complex learning and computation problems, the relation between the \weight threshold $\sigma$ and the \mwd $I_\sigma$. 
\mam{We conjecture that our general lower
bound could be improved for specific graph topologies, e.g., for (nearly-) uniform degree graphs.}

\bibliographystyle{plain} 
\bibliography{references,references2,references5}

@inproceedings{discversion,
  author    = {Dariusz R. Kowalski and
               Miguel A. Mosteiro},
  title     = {Weighted Beeping Networks},
  booktitle = {Proc. of the 40th International Symposium on Distributed Computing (DISC)},
  series       = {LIPIcs},
  publisher    = {Schloss Dagstuhl - Leibniz-Zentrum f{\"{u}}r Informatik},
  note      = {To appear.},
  year      = {2026}
}

@article{DILLENCOURT2025106516,
title = {Leveraging parameterized Chernoff bounds for simplified algorithm analyses},
journal = {Information Processing Letters},
volume = {187},
pages = {106516},
year = {2025},
issn = {0020-0190},
doi = {https://doi.org/10.1016/j.ipl.2024.106516},
url = {https://www.sciencedirect.com/science/article/pii/S0020019024000462},
author = {Michael Dillencourt and Michael T. Goodrich and Michael Mitzenmacher}
}

@inproceedings{dufoulon2020can,
  title={Can Uncoordinated Beeps tell Stories?},
  author={Dufoulon, Fabien and Burman, Janna and Beauquier, Joffroy},
  booktitle={Proceedings of the 39th Symposium on Principles of Distributed Computing},
  pages={408--417},
  year={2020}
}

@inproceedings{feldmann2020time,
  title={Time-and space-optimal discrete clock synchronization in the beeping model},
  author={Feldmann, Michael and Khazraei, Ardalan and Scheideler, Christian},
  booktitle={Proceedings of the 32nd ACM Symposium on Parallelism in Algorithms and Architectures},
  pages={223--233},
  year={2020}
}

@inproceedings{kleinberg2000small,
  title={The small-world phenomenon: An algorithmic perspective},
  author={Kleinberg, Jon},
  booktitle={Proceedings of the thirty-second annual ACM symposium on Theory of computing},
  pages={163--170},
  year={2000}
}

@inproceedings{garncarek2025deterministic,
  author       = {Pawel Garncarek and
                  Tomasz Jurdzinski and
                  Dariusz R. Kowalski and
                  Shay Kutten and
                  Miguel A. Mosteiro},
  IGNOREeditor       = {Ho{-}Lin Chen and
                  Wing{-}Kai Hon and
                  Meng{-}Tsung Tsai},
  title        = {Deterministic Local Problems in Radio Networks: On the Impact of Local
                  Domination and a Bit of Advice},
  booktitle    = {36th International Symposium on Algorithms and Computation, {ISAAC}
                  2025, Tainan, Taiwan, December 7-10, 2025},
  series       = {LIPIcs},
  IGNOREvolume       = {359},
  IGNOREpages        = {34:1--34:20},
  publisher    = {Schloss Dagstuhl - Leibniz-Zentrum f{\"{u}}r Informatik},
  year         = {2025},
  url          = {https://doi.org/10.4230/LIPIcs.ISAAC.2025.34},
  IGNOREdoi          = {10.4230/LIPICS.ISAAC.2025.34},
  bibsource    = {dblp computer science bibliography, https://dblp.org}
}

@inproceedings{lynch2015truly,
  title={A (truly) local broadcast layer for unreliable radio networks},
  author={Lynch, Nancy and Newport, Calvin},
  booktitle={Proceedings of the 2015 ACM Symposium on Principles of Distributed Computing (PODC'15)},
  pages={109--118},
  year={2015}
}

@article{Hay,
	author = {J.F. Hayes},
	title = {An adaptive technique for local distribution},
	journal = {IEEE Communications Magazine},
	issue_date = {March 1981},
	volume = {19},
	number = {2},
	year = {1981},
}

@article{ASHKENAZI2022104925,
title = {Noisy beeping networks},
journal = {Information and Computation},
volume = {289},
pages = {104925},
year = {2022},
issn = {0890-5401},
doi = {https://doi.org/10.1016/j.ic.2022.104925},
url = {https://www.sciencedirect.com/science/article/pii/S0890540122000785},
author = {Yagel Ashkenazi and Ran Gelles and Amir Leshem}
}

@inproceedings{gilbert2017symmetry,
  title={Symmetry Breaking with Noisy Processes},
  author={Gilbert, Seth and Newport, Calvin},
  booktitle={Proceedings of the ACM Symposium on Principles of Distributed Computing},
  pages={273--282},
  year={2017}
}

@inproceedings{efremenko2020noisy,
  title={Noisy beeps},
  author={Efremenko, Klim and Kol, Gillat and Saxena, Raghuvansh R},
  booktitle={Proceedings of the 39th Symposium on Principles of Distributed Computing},
  pages={418--427},
  year={2020}
}

@inproceedings{DBLP:conf/esa/GarncarekKKM25,
  author       = {Pawel Garncarek and
                  Dariusz R. Kowalski and
                  Shay Kutten and
                  Miguel A. Mosteiro},
  editor       = {Anne Benoit and
                  Haim Kaplan and
                  Sebastian Wild and
                  Grzegorz Herman},
  title        = {Beeping Deterministic {CONGEST} Algorithms in Graphs},
  booktitle    = {33rd Annual European Symposium on Algorithms, {ESA} 2025, Warsaw,
                  Poland, September 15-17, 2025},
  series       = {LIPIcs},
  pages        = {20:1--20:17},
  publisher    = {Schloss Dagstuhl - Leibniz-Zentrum f{\"{u}}r Informatik},
  year         = {2025},
  url          = {https://doi.org/10.4230/LIPIcs.ESA.2025.20},
  doi          = {10.4230/LIPICS.ESA.2025.20},
  bibsource    = {dblp computer science bibliography, https://dblp.org}
}

@article{oroojlooy2023review,
  title={A review of cooperative multi-agent deep reinforcement learning},
  author={Oroojlooy, Afshin and Hajinezhad, Davood},
  journal={Applied Intelligence},
  volume={53},
  number={11},
  pages={13677--13722},
  year={2023},
  publisher={Springer}
}

@inproceedings{mao2020neighborhood,
  title={Neighborhood cognition consistent multi-agent reinforcement learning},
  author={Mao, Hangyu and Liu, Wulong and Hao, Jianye and Luo, Jun and Li, Dong and Zhang, Zhengchao and Wang, Jun and Xiao, Zhen},
  booktitle={Proceedings of the AAAI conference on artificial intelligence},
  volume={34(5)},
  pages={7219--7226},
  year={2020}
}

@mastersthesis{malucellineighboring,
  title={Neighboring-based Strategies for Multi-Agent Reinforcement Learning},
  author={Malucelli, Nicol{\`o}},
school={Universita di Bologna},
year={2024}
}

@article{yang2024neighborhood,
  title={Neighborhood curiosity-based exploration in multi-agent reinforcement learning},
  author={Yang, Shike and He, Ziming and Li, Jingchen and Shi, Haobin and Ji, Qingbing and Hwang, Kao-Shing and Li, Xianshan},
  journal={IEEE Transactions on Cognitive and Developmental Systems},
  year={2024},
  publisher={IEEE}
}

@article{chen2024marl,
  title={MARL-LNS: Cooperative Multi-agent Reinforcement Learning via Large Neighborhoods Search},
  author={Chen, Weizhe and Koenig, Sven and Dilkina, Bistra},
  journal={arXiv preprint arXiv:2404.03101},
  year={2024}
}

@article{barabasi1999emergence,
  title={Emergence of scaling in random networks},
  author={Barab{\'a}si, Albert-L{\'a}szl{\'o} and Albert, R{\'e}ka},
  journal={science},
  volume={286},
  number={5439},
  pages={509--512},
  year={1999},
  publisher={American Association for the Advancement of Science}
}

@article{bansal2007infer,
     title={How to infer gene networks from expression profiles},
     author={Bansal, Mukesh and Belcastro, Vincenzo and Ambesi-Impiombato, Alberto and Di Bernardo, Diego},
     journal={Molecular systems biology},
     volume={3},
     number={1},
     year={2007},
     publisher={Wiley Online Library}
}

@inproceedings{davies2023optimal,
  title={Optimal message-passing with noisy beeps},
  author={Davies, Peter},
  booktitle={Proceedings of the 2023 ACM Symposium on Principles of Distributed Computing},
  pages={300--309},
  year={2023}
}

@string{isaac = {Intl. Symp. on Algorithms and Computation}}

@string{lncs = {Lecture Notes in Computer Science}}

@string{svb = {Springer-Verlag, Berlin}}

@inproceedings{HWaff,
	Author = {Magn{\'u}s M. Halld{\'o}rsson and Roger Wattenhofer},
	Booktitle = {Proc. of the 36th International Colloquium on Automata, Languages and Programming (ICALP)},
	Pages = {525--536},
	Title = {Wireless Communication is in APX},
	Year = {2009}}

@inproceedings{asgeirsson2017universal,
  title={Universal framework for wireless scheduling problems},
  author={{\'A}sgeirsson, Eyj{\'o}lfur I and Halld{\'o}rsson, Magn{\'u}s M and Tonoyan, Tigran},
  booktitle={44th International Colloquium on Automata, Languages, and Programming (ICALP 2017)},
  year={2017},
  organization={Schloss Dagstuhl-Leibniz-Zentrum fuer Informatik}
}

@article{dong2023social,
  title={Social signal learning of the waggle dance in honey bees},
  author={Dong, Shihao and Lin, Tao and Nieh, James C and Tan, Ken},
  journal={Science},
  volume={379},
  number={6636},
  pages={1015--1018},
  year={2023},
  publisher={American Association for the Advancement of Science}
}

@book{HWsurvey2019,
author="Halld{\'o}rsson, Magn{\'u}s M.
and Wattenhofer, Roger",
editor="Steffen, Bernhard
and Woeginger, Gerhard",
title="Wireless Network Algorithmics",
bookTitle="Computing and Software Science: State of the Art and Perspectives",
year="2019",
publisher="Springer International Publishing",
address="Cham",
pages="141--160",
isbn="978-3-319-91908-9",
NooPdoi="10.1007/978-3-319-91908-9_9",
NooPurl="https://doi.org/10.1007/978-3-319-91908-9_9"
}

@article{liu2024review,
  title={A Review of Graph Neural Networks in Epidemic Modeling},
  author={Liu, Zewen and Wan, Guancheng and Prakash, B Aditya and Lau, Max SY and Jin, Wei},
  journal={arXiv preprint arXiv:2403.19852},
  year={2024}
}

@article{FKVaffJournal,
author="Fangh{\"a}nel, Alexander
and Kesselheim, Thomas
and V{\"o}cking, Berthold",
title = "Improved algorithms for latency minimization in wireless networks",
journal = "Theoretical Computer Science",
volume = "412",
number = "24",
pages = "2657 - 2667",
year = "2011"
}

@inproceedings{FKVaff,
author="Fangh{\"a}nel, Alexander
and Kesselheim, Thomas
and V{\"o}cking, Berthold",
title="Improved Algorithms for Latency Minimization in Wireless Networks",
bookTitle="36th International Colloquium on Automata, Languages and Programming (ICALP)",
year="2009",
publisher="Springer Berlin Heidelberg",
pages="447--458"
}

@inproceedings{Kaff,
	Author = {Thomas Kesselheim},
	Booktitle = {Proc. of the 31st Annual ACM SIGACT-SIGOPS Symposium on Principles of Distributed Computing (PODC)},
	Pages = {281--290},
	Title = {Dynamic packet scheduling in wireless networks},
	Year = {2012}}

@inproceedings{KVaff,
	Author = {Thomas Kesselheim and Berthold V{\"o}cking},
	Booktitle = {Proc. of the 24th International Symposium on Distributed Computing (DISC)},
	Pages = {163--178},
	Publisher = svb,
	Series = lncs,
	Title = {Distributed Contention Resolution in Wireless Networks},
	Volume = 6343,
	Year = {2010}}

@inproceedings{KMR_fomc14,
  author    = {Dariusz R. Kowalski and
               Miguel A. Mosteiro and
               Tevin Rouse},
  title     = {Dynamic multiple-message broadcast: bounding throughput in the affectance
               model},
  booktitle = {10th {ACM} International Workshop on Foundations of Mobile Computing,
               {FOMC} 2014, Philadelphia, PA, USA, August 11, 2014},
  pages     = {39--46},
  year      = {2014}
}

@article{KowalskiMZjournal,
  author       = {Dariusz R. Kowalski and
                  Miguel A. Mosteiro and
                  Kevin Zaki},
  title        = {Dynamic Multiple-Message Broadcast: Bounding Throughput in the Affectance
                  Model},
  journal      = {Theory Comput. Syst.},
  volume       = {67},
  number       = {4},
  pages        = {825--854},
  year         = {2023},
  url          = {https://doi.org/10.1007/s00224-023-10131-1},
  doi          = {10.1007/S00224-023-10131-1},
  bibsource    = {dblp computer science bibliography, https://dblp.org}
}

@inproceedings{DeMarcoK25,
  author       = {Gianluca {De Marco} and
                  Dariusz R. Kowalski},
  OLDeditor       = {Keren Censor{-}Hillel and
                  Fabrizio Grandoni and
                  Jo{\"{e}}l Ouaknine and
                  Gabriele Puppis},
  title        = {Ultra-Resilient Superimposed Codes: Near-Optimal Construction and
                  Applications},
  booktitle    = {52nd International Colloquium on Automata, Languages, and Programming,
                  {ICALP} 2025, July 8-11, 2025, Aarhus, Denmark},
  series       = {LIPIcs},
  volume       = {334},
  pages        = {65:1--65:20},
  OLDpublisher    = {Schloss Dagstuhl - Leibniz-Zentrum f{\"{u}}r Informatik},
  year         = {2025},
  OLDurl          = {https://doi.org/10.4230/LIPIcs.ICALP.2025.65},
  OLDdoi          = {10.4230/LIPICS.ICALP.2025.65},
  bibsource    = {dblp computer science bibliography, https://dblp.org}
}

@article{moreno2023quorum,
  title={Quorum sensing as a mechanism to harness the wisdom of the crowds},
  author={Moreno-G{\'a}mez, Stefany and Hochberg, Michael E and Van Doorn, GS},
  journal={Nature communications},
  volume={14},
  number={1},
  pages={3415},
  year={2023},
  publisher={Nature Publishing Group UK London}
}

@article{miller2001quorum,
  title={Quorum sensing in bacteria},
  author={Miller, Melissa B and Bassler, Bonnie L},
  journal={Annual Reviews in Microbiology},
  volume={55},
  number={1},
  pages={165--199},
  year={2001},
  publisher={Annual Reviews 4139 El Camino Way, PO Box 10139, Palo Alto, CA 94303-0139, USA}
}

@article{fuqua1994quorum,
  title={Quorum sensing in bacteria: the LuxR-LuxI family of cell density-responsive transcriptional regulators},
  author={Fuqua, W Claiborne and Winans, Stephen C and Greenberg, E Peter},
  journal={Journal of bacteriology},
  volume={176},
  number={2},
  pages={269--275},
  year={1994}
}

@article{waters2005quorum,
  title={Quorum sensing: cell-to-cell communication in bacteria},
  author={Waters, Christopher M and Bassler, Bonnie L},
  journal={Annu. Rev. Cell Dev. Biol.},
  volume={21},
  number={1},
  pages={319--346},
  year={2005},
  publisher={Annual Reviews}
}

@inproceedings{nr,
  title={The network data repository with interactive graph analytics and visualization},
  author={Rossi, Ryan and Ahmed, Nesreen},
  booktitle={Proceedings of the AAAI conference on artificial intelligence},
  volume={29},
  number={1},
  year={2015},
  note={\url{https://networkrepository.com}}
}

@article{firth2015experimental,
  title={Experimental manipulation of avian social structure reveals segregation is carried over across contexts},
  author={Firth, Josh A and Sheldon, Ben C},
  journal={Proceedings of the Royal Society B: Biological Sciences},
  volume={282},
  number={1802},
  pages={20142350},
  year={2015},
  publisher={The Royal Society}
}

@article{mersch2013ant,
  title={Tracking individuals shows spatial fidelity is a key regulator of ant social organization},
  author={Mersch, Danielle P and Crespi, Alessandro and Keller, Laurent},
  journal={Science},
  volume={340},
  number={6136},
  pages={1090--1093},
  year={2013},
  publisher={American Association for the Advancement of Science}
}

@article{gazda2015importance,
  title={The importance of delineating networks by activity type in bottlenose dolphins (Tursiops truncatus) in Cedar Key, Florida},
  author={Gazda, Stefanie and Iyer, Swami and Killingback, Timothy and Connor, Richard and Brault, Solange},
  journal={Royal Society Open Science},
  volume={2},
  number={3},
  pages={140263},
  year={2015},
  publisher={The Royal Society Publishing}
}

@book{MitzenmacherUpfal-book2005,
author	={Michael Mitzenmacher and Eli Upfal},
title	={Probability and Computing},
publisher={Cambridge University Press},
year	={2005}
}

@article{galperin2004bacterial,
  title={Bacterial signal transduction network in a genomic perspective},
  author={Galperin, Michael Y},
  journal={Environmental microbiology},
  volume={6},
  number={6},
  pages={552--567},
  year={2004},
  publisher={Wiley Online Library}
}

@article{hounkanli2020global,
  title={Global synchronization and consensus using beeps in a fault-prone multiple access channel},
  author={Hounkanli, Kokouvi and Miller, Avery and Pelc, Andrzej},
  journal={Theoretical Computer Science},
  volume={806},
  pages={567--576},
  year={2020},
  publisher={Elsevier}
}

@inproceedings{dufoulon2018beeping,
  title={Beeping a deterministic time-optimal leader election},
  author={Dufoulon, Fabien and Burman, Janna and Beauquier, Joffroy},
  booktitle={32nd International Symposium on Distributed Computing (DISC 2018)},
  year={2018},
  organization={Schloss Dagstuhl-Leibniz-Zentrum fuer Informatik}
}

@inproceedings{forster2014deterministic,
  title={Deterministic leader election in multi-hop beeping networks},
  author={F{\"o}rster, Klaus-Tycho and Seidel, Jochen and Wattenhofer, Roger},
  booktitle={Distributed Computing: 28th International Symposium, DISC 2014, Austin, TX, USA, October 12-15, 2014. Proceedings 28},
  pages={212--226},
  year={2014},
  organization={Springer}
}

@article{afek2013beeping,
  title={Beeping a maximal independent set},
  author={Afek, Yehuda and Alon, Noga and Bar-Joseph, Ziv and Cornejo, Alejandro and Haeupler, Bernhard and Kuhn, Fabian},
  journal={Distributed computing},
  volume={26},
  number={4},
  pages={195--208},
  year={2013},
  publisher={Springer}
}

@inproceedings{halldorsson2012towards,
  title={Towards tight bounds for local broadcasting},
  author={Halld{\'o}rsson, Magn{\'u}s M and Mitra, Pradipta},
  booktitle={Proceedings of the 8th International Workshop on Foundations of Mobile Computing},
  pages={1--9},
  year={2012}
}

@inproceedings{cornejo2010deploying,
  title={Deploying wireless networks with beeps},
  author={Cornejo, Alejandro and Kuhn, Fabian},
  booktitle={Distributed Computing: 24th International Symposium, DISC 2010, Cambridge, MA, USA, September 13-15, 2010. Proceedings 24},
  pages={148--162},
  year={2010},
  organization={Springer}
}

@article{afek2011biological,
  title={A biological solution to a fundamental distributed computing problem},
  author={Afek, Yehuda and Alon, Noga and Barad, Omer and Hornstein, Eran and Barkai, Naama and Bar-Joseph, Ziv},
  journal={Science},
  volume={331},
  number={6014},
  pages={183--185},
  year={2011},
  publisher={American Association for the Advancement of Science}
}

\end{document}